\documentclass{article}
\usepackage{wright}

\usepackage[many]{tcolorbox}

\tcbset{
  algorithmbox/.style={
    breakable,
    enhanced,
    colback=black!2,
    colframe=black!50,
    colbacktitle=black!7,
    coltitle=black,
    coltext=black,
    boxrule=0.5pt,
    arc=0.5pt,
    left=6pt,
    right=6pt,
    top=5pt,
    bottom=5pt,
    before skip=10pt,
    after skip=10pt
  }
}

\newtcolorbox[
  auto counter,
  crefname={Algorithm}{Algorithms}
]{myalgorithm}[2][]{
  algorithmbox,
  fonttitle=\bfseries,
  title={Algorithm~\thetcbcounter: #2},
  #1
}

\title{Pretty good shadow tomography}
\author{Sabee Grewal\thanks{IBM Quantum, Columbia University. \texttt{sabee@ibm.com}}  \and Angelos Pelecanos\thanks{UC Berkeley. \texttt{\{apelecan,jspilecki,jswright\}@berkeley.edu}} \and Jack Spilecki\footnotemark[2] \and Ewin Tang\thanks{Princeton. \texttt{ent@cs.princeton.edu}} \and John Wright\footnotemark[2]}

\date{}

\begin{document}

\maketitle

\begin{abstract}
We give a shadow tomography algorithm that uses $O(\log(m) \log(m/\delta) / \varepsilon^2)$ copies of the input state. 
Specifically, given a quantum state $\rho$ and observables $O_1,\ldots,O_m$, the algorithm outputs estimates of $\tr(O_i\rho)$ for all $i\in[m]$ that are all accurate to within $\pm\varepsilon$, except with probability at most $\delta$.
The key idea is to work in the average-case setting: for a random state drawn from a known prior, a pretty good measurement on $O(\log(m)/\eps^2)$ copies gives a constant-error guarantee on all but a small fraction of the prior. 
Boosting this guarantee gives an average-case shadow tomography algorithm, and the minimax principle then (non-constructively) implies the existence of a measurement with the same worst-case guarantee.

We also identify an obstruction to further improvement with this approach.
While pretty good measurements amplify success over the randomness of the prior, they do not amplify success over the randomness of measurement.
This issue appears to be inherent, and gives rise to the remaining $\log(m)$ factor separating the complexity of the quantum problem from its classical counterpart.
\end{abstract}

\vfill
{\small
\noindent
\textbf{A.I.\ disclosure.}
The main result was developed through heavy interaction with generative A.I.\ models, with the proofs generated by the models themselves; the analysis presented here was modified to add conceptual clarity.
See \cref{subsec:ai} for more details and additional context for this result.
}
\vspace{2em}

\newpage

\hypersetup{linktocpage}
\tableofcontents
\thispagestyle{empty}

\newpage

\section{Introduction}

Measuring quantum states is delicate business.
Observing a state necessarily disturbs it, so if we measure where it lies along one axis, we destroy the information of where it lies along a different axis.
This central property of quantum mechanics is known as \emph{measurement uncertainty}, and like Heisenberg uncertainty, it arises because observables need not commute with each other.
Unlike in the classical world, where observables commute happily, it matters in which order we apply measurements to quantum states, as states degrade with each successive measurement; we grapple with this issue every time we try to extract information from quantum states.

Measurement uncertainty has been formalized many times across almost a century of literature.
Early literature focuses on quantifying the tradeoff between how much information one gains from a measurement and how much disturbance it causes, starting with Heisenberg~\cite{Hei27}\footnote{Heisenberg originally described his principle in terms of how a measurement in one basis disturbs the state in another; what we know now as the Heisenberg uncertainty principle is a later refinement, capturing the slightly different phenomenon that quantum states cannot be ``certain'' (low standard deviation) with respect to two different axes, and was shown by Kennard~\cite{Ken27}, Robertson~\cite{Rob29}, and Schr\"odinger~\cite{Sch30}.} and progressively sharpened over time~\cite{AK65,AG88,FP96,Oza03,BLW13,BHOW14}; see also the review~\cite{BLW14} and the closely-related gentle measurement lemma~\cite{Win99,ON02}.
More recently, researchers have investigated more operational variants of this question, giving bounds on the quality of (unbiased) estimators for the expectations of two non-commuting observables from measuring one state~\cite{Nag89,Hay99}; see~\cite{CSLA21} for a recent exposition. 
However, essentially all of this literature focuses on the single-copy setting.

Since measurement uncertainty is a fundamental limit on single-copy measurement, we might expect the same phenomenon, and its operational consequences, to persist in the many-copy setting. 
The natural question to pose here is: does estimating non-commuting observables require more copies than estimating commuting ones?
Put another way, what is the cost of non-commutativity of measurement?
One might even expect the effect to be more than constant, an asymptotically worse statistical rate resulting from the inability to reuse states after measurement.
Little is known about this question.
Later, we discuss formulations of this question with a more statistical flavor; but in the quantum computing literature, this question is captured by the problem of \emph{shadow tomography}, posed by Aaronson in 2016~\cite{Aar16}. 

In the problem of shadow tomography, we are given copies of an unknown state $\rho$, along with some observables $O_1, \dots, O_m$.
Our goal is to estimate every expectation $\tr(O_i \rho)$ to within $\pm\eps$ accuracy with a total failure probability of at most $\delta$.
When $\rho$ is a classical distribution, there is no issue: we can reuse the same samples to estimate every observable, requiring only $O(\log(m / \delta)/\eps^2)$ samples; the $\log(m)$ comes from a union bound.
For a general quantum state, however, you run into immediate issues with non-commutativity: measuring one observable can disturb the state, preventing the measurement of the rest of them.
A naive algorithm, therefore, uses $O(m\log(m/\delta)/\eps^2)$ copies, while running full tomography requires $O(d^2/\eps^2)$ copies.
Aaronson showed that, surprisingly, shadow tomography can be solved using a number of copies that scales polynomially in $1/\eps$, $\log(m)$, and $\log(d)$~\cite{Aar18}!

Since then, subsequent work has steadily brought the sample complexity down~\cite{AR19,BO24}, until the latest work of~\cite{COPW26}, which uses only $O(\log(m)\sqrt{\log(d)}/\eps^3)$ copies.
This complexity is essentially optimal for the general approach taken by all prior work: designing sufficiently gentle measurements that reveal useful information while disturbing the state as little as possible, so that it can be reused in later rounds.
The persistent dependence on dimension suggested that this dependence might be necessary for shadow tomography, and perhaps reflects a fundamental cost of non-commutativity.
However, as we do not have any lower bound for shadow tomography stronger than $\Omega(\log(m/\delta)/\eps^2)$, the question remains open.

\subsection{Results}

We make further progress on shadow tomography with a new dimension-independent algorithm.


\begin{restatable}{theorem}{WorstCaseShadow}
\label{thm:intro-worst}
There is an algorithm for the shadow tomography problem that consumes 
\[
O\!\left(\frac{\log(m)\log(m/\delta)}{\eps^2}\right)
\]
copies of the input state.
\end{restatable}


Notably, this algorithm has no dependence on the dimension of the input state.
It gives the first dimension-free shadow tomography algorithm with only polylogarithmic dependence on $m$; prior to this work, the best dimension-free bound was $O(\sqrt{m}/\eps^2)$ copies, due to~\cite{COPW26}, improving an earlier $O(\sqrt{m}\log(m)/\eps^2)$ bound of Sinha~\cite{Sin25}.
With this, the best-known bounds for shadow tomography, with constant $\delta$, can be summarized as follows:

\begin{enumerate}
    \item $O(\log(m)/\eps^2)$ copies suffice in the high-precision regime $\eps \lesssim 1/d$~\cite{PSW25,PSTW25}. 

    \item $O(\log(m)\sqrt{\log(d)}/\eps^3)$ copies suffice~\cite{COPW26}; for $\eps \gtrsim 1/d$, this gives the best known bound when $\sqrt{\log(d)}/\eps \lesssim \log(m)$. 
    
    \item $O(\log^2(m)/\eps^2)$ copies suffice by this work; for $\eps \gtrsim 1/d$, this gives the best known bound when $\sqrt{\log(d)}/\eps \gtrsim \log(m)$. 
\end{enumerate}

Our key idea is to depart from previous work and study shadow tomography in the \emph{average case}, where the input state is drawn from a known prior.
Let us first state this average-case variant more formally. 
Suppose that the unknown state is promised to belong to some known family of states $\calF = \{\rho_1, \dots, \rho_\ell\}$, and let $p = (p_1, \dots, p_\ell)$ be a known prior on this family.
The input state $\brho = \rho_{\bx}$ is obtained by drawing $\bx \sim p$, and we are given copies of $\brho$ together with observables $O_1, \dots, O_m$.

As before, our goal is to output estimates of the expectation values $\tr(O_i \brho)$ for every $i \in [m]$ to accuracy $\pm \eps$. 
The key difference is that we no longer require the algorithm to succeed for every possible input state in $\calF$. 
Instead, we require the algorithm to succeed with probability at least $1-\delta$ over both the random choice of $\brho$ from the prior and the randomness of the measurement. 

Working in the average case gives us access to tools that are naturally defined with respect to a prior.
Most importantly, we can use (variants of) the \emph{pretty good measurement} (PGM), a canonical measurement first designed for the task of quantum state discrimination. 
With this, we can attain a shadow tomography algorithm in the average-case setting.

\begin{theorem}
\label{thm:intro-avg}
There is an algorithm for the average-case shadow tomography problem that consumes 
\[
O\!\left(\frac{\log(m)\log(m/\delta)}{\eps^2}\right)
\]
copies of the input state.
\end{theorem}

\Cref{thm:intro-worst} follows directly from \cref{thm:intro-avg} by Sion's minimax theorem~\cite{Sio58} (which generalizes the classical minimax theorem of von Neumann~\cite{vN28}).
We note that this approach is fundamentally non-constructive: the minimax theorem merely tells us that an algorithm exists, without giving an explicit description of one.\footnote{
The minimax measurement can in principle be found constructively by solving a semidefinite program over POVMs. 
This does not give an especially illuminating algorithm, however.
}
Indeed, this non-constructiveness may be the fundamental strength of this approach. 
As we will discuss later, a major obstacle to designing shadow tomography algorithms is that we do not understand what a hard instance looks like for any given set of observables.
The average-case formulation lets us avoid this problem entirely: rather than identifying the hardest prior, we prove that our algorithm works for \emph{every} prior, allowing the measurement to depend on the prior it is given.
The minimax theorem then turns these average-case guarantees into a worst-case guarantee without ever requiring us to identify which prior is hardest.

However, this approach appears to be bottlenecked at $\log^2(m)$.
In the next section, we will discuss why, and in the process detail our current understanding of pretty good measurements and average-case shadow tomography.

\paragraph{Concurrent work.}
Concurrent to our work, Jeronimo, Huang, and Liu attained the same complexity bound, going through similar techniques, though with a more intricate algorithm and analysis.
Though our proof is AI-generated, we can be fairly confident that it was generated independently: we first produced a shadow tomography algorithm using $O(\log^2(m)\log(m/\delta)/\eps^2)$ copies on July 28th and the algorithm presented here on August 2nd.
Jeronimo, Huang, and Liu posted their paper to arXiv on August 6th, initially attaining a $\widetilde{O}(\log^4(m/\delta)/\eps^2)$ shadow tomography algorithm, and subsequently improved the algorithm to the $O(\log(m)\log(m/\delta)/\eps^2)$ algorithm, posted on September 14th.

\paragraph{Related work on the cost of noncommutativity.}

Another natural way to investigate the cost of non-commutativity is with a statistical lens: given observables $O_1, \dots, O_m$ and $n$ copies of the state $\rho$, compute unbiased estimators $\bo_1,\dots,\bo_m$ of $\tr(O_1 \rho), \dots, \tr(O_m \rho)$ such that the total variance $\sum_i \Var[\bo_1] + \dots + \Var[\bo_m]$ is as small as possible.
In particular, the question is about the rate: as $n \to \infty$, the total variance will scale as some $C/n + o(1/n)$; the goal is to identify $C$, the dominant contribution to the rate.
For this question, results indicate that the cost of non-commutativity is at most a factor of $2$: while this variance can be as large as $m$ for commuting observables, it is at most $2m$ for non-commuting observables.
This factor of two was first described by Nagaoka and Hayashi in the single-copy setting~\cite{Nag89,Hay99} and later extended to various multi-copy settings~\cite{TAD20,CSDV19,ABR26}.
In particular, this gap is the same factor of two which separates the Quantum Cram\'{e}r--Rao bound~\cite{Hel67} from the Holevo--Cram\'{e}r--Rao bound~\cite{Hol11}, and is a direct consequence of the \emph{measurement incompatibility} that arises from non-commutativity.
This factor of two is tight; see, for example, \cite[Section 3.1.1]{DGG20}.

For these formulations, the central focus is on the variance, for which the pretty-good measurement has especially nice behavior (see \Cref{subsec:pgm-mse}).
By comparison, shadow tomography demands control of higher moments, since the desired $\log(m)$ scaling comes from demonstrating an exponential tail on the estimator distribution.
This makes the two formulations substantially different: even though the cost of non-commutativity is a constant factor for variances, it need not be the same for shadow tomography.\footnote{
Here, we have restricted our discussion to allowing for general measurements.
In restricted settings, such as requiring measurements to only operate on single copies, the cost of non-commutativity is essentially $m$: it is near-optimal to measure each observable on $O(\log(m)/\eps^2)$ copies of $\rho$, without any reuse~\cite{CCHL22}.
}

\section{Background}
\label{background}

In this section, we will describe the broad approach that we took when designing our shadow tomography algorithm, and what motivated us to consider this approach.
We will then survey several proof directions that we attempted, some of which were successful, some of which were not, and we will conclude by describing the high level proof strategy that we have ultimately decided to take.

\subsection{The pretty good measurement in full state tomography}

To motivate the approach we take to the shadow tomography problem, let us first consider the more challenging task of full state tomography.
Suppose one is given $n$ copies of an unknown state, where the state is promised to be an element of some family $\calF = \{\rho_1, \ldots, \rho_{\ell}\}$ of states.
For ease of exposition, we will assume that this family is of finite size, but in practice the state families considered in full state tomography (e.g.\ the set of all pure states) are usually infinite and continuous.\footnote{We note, however, that it always suffices to consider finite-size families, as continuous families of states can always be discretized.}
If $\rho_i$ is the unknown state, a tomography algorithm is allowed to perform any measurement on the input $\rho_i^{\otimes n}$, and its goal is to output an estimator $\widehat{\brho}$ which is close to $\rho_i$ in trace distance (with, say, 99\% probability).

There is a common two-step strategy for designing full state tomography algorithms, often achieving optimal sample complexities.
\begin{enumerate}
\item First, identify a natural ``hard'' distribution $p = (p_1, \ldots, p_{\ell})$ over the states of this family.
Ideally, solving tomography over a state drawn from this distribution should be every bit as hard as solving tomography over a worst-case state drawn from the family $\calF$.
\item Second, define the tomography algorithm to be the \emph{Pretty Good Measurement (PGM)} over this distribution of states. This is the measurement $M = \{M_1, \ldots, M_{\ell}\}$ defined as
\begin{equation}\label{eq:pgm-def}
    M_i = S^{-1/2} \cdot p_i \rho_i^{\otimes n} \cdot S^{-1/2}, \text{ for all $i \in [\ell]$,}
    \qquad \text{where } S = p_1 \rho_1^{\otimes n} + \cdots + p_{\ell} \rho_{\ell}^{\otimes n}.
\end{equation}
\end{enumerate}
In other words, given the input $\rho_i^{\otimes n}$, the algorithm will perform the measurement $M$, receive some outcome $\bj \in [\ell]$, and output the estimator $\widehat{\brho} = \rho_{\bj}$, which is hopefully close to $\rho_i$ in trace distance.

The PGM originates from the literature on the related problem of \emph{quantum state discrimination}.
An instance of the average-case variant of quantum state discrimination is specified by a set $\calL$ of labels $x$, a set of states $\{\sigma_x\}_{x \in \calL}$ labeled by $\calL$, and a distribution $q = (q_x)_{x \in \calL}$ on the states.
It is convenient to take the set of labels to be $\calL = \{1, \ldots, \ell\}$,
but later we will sometimes view the labels as specifying a real-valued parameter of the states that we would like to compute.
Given a state $\sigma_{\bx}$ drawn from this distribution, the goal is to output $\widehat{\bx} \in \calL$ such that $\widehat{\bx} = \bx$.
Hence, viewing $\bx$ as a parameter of the state $\sigma_{\bx}$, quantum state discrimination is the difficult task of \emph{exactly} computing this parameter (whereas in full state tomography, we are only asked to approximate a given parameter).
The algorithm's success probability is given by its average case success probability over the different input states, i.e.\
\begin{equation*}
    \E_{\bx \sim q}\big[\Pr[\widehat{\bx} = \bx \mid \sigma_{\bx}]\big].
\end{equation*}
If $P_{\mathrm{OPT}}(q)$ is the success probability of the optimal strategy and $\overline{P}_{\mathrm{OPT}}(q) = 1 - P_{\mathrm{OPT}}(q)$ is its failure probability (and similarly for $P_{\mathrm{PGM}}(q)$ and $\overline{P}_{\mathrm{PGM}}(q)$),
then Barnum and Knill~\cite{BK02} showed that $P_{\mathrm{PGM}}(q) \geq P_{\mathrm{OPT}}(q)^2$,
which in turn implies that
\begin{equation}\label{eq:pretty-good}
    \overline{P}_{\mathrm{PGM}}(q) \leq 2 \cdot \overline{P}_{\mathrm{OPT}}(q).
\end{equation}
This gives a formal sense in which the PGM is ``pretty good'' and justifies the PGM's role as the canonical default measurement in quantum state discrimination.

There is no analogous justification for why the PGM should serve the same role in full state tomography, yet in practice many algorithms are designed with the above two-step strategy.
Examples include Hayashi's algorithm for pure state tomography~\cite{Hay98}, in which the chosen distribution is the Haar measure over pure states;
the tomography algorithm of~\cite{SSW25} for rank-$r$ projector states (i.e.\ states of the form $\rho = P/r$, where $P$ is a rank-$r$ projector),  in which the chosen distribution is the Haar measure over projector states;
the algorithm of~\cite{PSTW25} for general mixed state tomography based on the random purification channel, in which the chosen distribution is the Hilbert-Schmidt measure over mixed states;
and one of the two mixed state tomography algorithms of~\cite{HHJ+16}.
In all these cases, the families of states being learned are highly symmetric, allowing one to choose distributions on these states which are themselves highly symmetric. This leads to tomography algorithms which are elegant to describe and clean to analyze.

As an aside, let us note one more property of the PGM in quantum state discrimination.
If the states $\sigma_x$ \emph{are} actually tensor product states, i.e.\ $\sigma_x = \tau_x^{\otimes n}$, for some set of states $\{\tau_x\}_{x \in \calL}$, then~\cite{HW12} observed that the PGM's failure probability can always be bounded by
\begin{equation}\label{eq:pgm-fail}
    \overline{P}_{\mathrm{PGM}}(q) \leq |\calL| \cdot F^n,
\end{equation}
where $F$ is the maximum pairwise fidelity between two distinct states $\tau_x \neq \tau_y$.
Crucially, this means that in the setting of quantum state discrimination, the PGM's failure probability \emph{decays exponentially} in the number of copies of the state one is given so long as this maximum fidelity is bounded away from 1, say $F \leq 1 - \epsilon$, a property we will revisit throughout this section.

\subsection{Hard distributions in shadow tomography}
\label{subsec:background-minimax}

Now let us attempt to design a shadow tomography algorithm using this strategy.
Let $O_1, \ldots, O_m$ be the observables we would like to estimate,
and let $\calF = \{\rho_1, \ldots, \rho_k\}$ be the finite family of quantum states that our unknown state is promised to be drawn from.
The first step in this strategy is to identify a natural ``hard'' distribution $p = (p_1, \ldots, p_k)$ over the states in $\calF$ for which the task of shadow tomography is every bit as hard as it is for a worst case state drawn from $\calF$.

What should this hard distribution look like?
It is perhaps difficult to answer this in a precise, formal sense, but we can list a couple of properties that it should intuitively satisfy.
First, if $\brho$ is a state drawn from this  distribution, then for every observable $O_i$, the observable value $\tr(O_i \cdot \brho)$ should have a reasonably large variance over the randomness in the state $\brho$.
If not, then one could predict $\tr(O_i \cdot \brho)$  by ignoring the input and outputting a fixed value independent of the state.
Second, for any $i \neq j$, we should want different observable values $\tr(O_i \cdot \brho)$ and $\tr(O_j \cdot \brho)$ to vary around their means in uncorrelated ways.
If instead observable values are correlated, then correctly predicting one observable would allow one to correctly predict another one, without having to take the second observable into account when designing one's measurement.

Designing distributions which satisfy these properties seems challenging, and to date, there is no canonical method to specify a hard distribution for a given set of observables.
Because of how multiple observables interact, hard distributions can be highly nonsymmetric and challenging to describe elegantly.
Surprisingly, this is still true when the observables themselves are highly structured, like the Pauli matrices $\{I, X, Y, Z\}^{\otimes k}$.

\paragraph{Worst-case quantum state discrimination.}
A similar issue has been faced before, in the literature on the \emph{worst}-case variant of quantum state discrimination.
An instance of this problem is specified only by a set of states $\{\sigma_x\}_{x \in \calL}$ (without an additional probability distribution).
Given a state $\sigma_{x}$ in this set, the goal is to output an $\widehat{\bx} \in \calL$ such that $\widehat{\bx} = x$.
The algorithm's success probability is given by its worst-case success probability over the different input states, i.e.\
\begin{equation*}
    \min_{x \in \calL}\big\{\Pr[\widehat{\bx} = x \mid \sigma_x]\big\}.
\end{equation*}
A natural way to design an algorithm for this problem is to identify a ``hard'' distribution $q = (q_1, \ldots, q_k)$ on these states and consider the PGM on this distribution. 
Exactly this approach was taken by Montanaro in the work~\cite{Mon19}, where he considered the performance of the PGM for worst-case state discrimination when defined on the uniform distribution $q = (\frac{1}{|\calL|}, \ldots, \frac{1}{|\calL|})$.
This is clearly not the \emph{hardest} distribution for most sets of states, but Montanaro's results showed that it was still ``hard enough'' to rederive some known bounds on worst-case state discrimination.
However, if one wants to build an algorithm based on the truly hardest choice of distribution $q$, then again one seems stuck, as we do not know a clean way to express exactly what this distribution is.

However, Harrow and Winter~\cite{HW12} observed a simple way to get around this issue.
They point out that a direct application of the minimax theorem shows that if $P_{\mathrm{OPT}}$ is the optimal probability of success in the worst-case setting, then
\begin{equation*}
P_{\mathrm{OPT}} = \min_q\big\{P_{\mathrm{OPT}}(q)\big\}.
\end{equation*}
Hence, to lower bound the optimal success probability in the worst case,
it suffices to lower bound the optimal success probability in the average case, for all prior distributions $q$.
This eliminates the need to identify a candidate hard distribution: instead, one is directly given a distribution $q$ and can design an algorithm specifically for this distribution $q$.
Using this, Harrow and Winter were able to take average-case lower bounds on the PGM shown by Barnum and Knill~\cite{BK02} and port them over to the worst-case setting.

\paragraph{Minimax for shadow tomography.}
Inspired by this, let us consider an average-case variant of shadow tomography in which the input state $\brho$, rather than being a worst-case element of the family $\calF$, is sampled from a known prior $p = (p_1, \ldots, p_{\ell})$ on the states in $\calF$.
Let us write $\bo = (\bo_1, \ldots, \bo_m)$ for its $m$ observable values, where $\bo_i = \tr(O_i \cdot \brho)$.
A shadow tomography algorithm can be specified by a POVM $M = \{M_{\widehat{o}}\}_{\widehat{o}}$, where the measurement outcomes $\widehat{o} = (\widehat{o}_1, \ldots, \widehat{o}_m) \in \R^m$ are the predictions for the $m$ observable values.
Its success probability can be written as
\begin{equation*}
    P_{\mathrm{succ}}(p, M) \coloneqq \Pr_{\brho, \widehat{\bo}}\big[\text{for all }i \in [m], \,|\widehat{\bo}_i - \bo_i| \leq \eps\big],
\end{equation*}
where $\widehat{\bo}$ is the result of measuring $\brho^{\otimes n}$ with $M$.
The set of all probability distributions $p$ is convex, as is the set of all POVMs $M$, and $\mathrm{P}_{\mathrm{succ}}(p, M)$ is linear in both $p$ and $M$.
As a result, the minimax theorem tells us that
\begin{equation*}
    \max_M \min_p \{P_{\mathrm{succ}}(p, M)\} = \min_p \max_M \{P_{\mathrm{succ}}(p, M)\}.
\end{equation*}
The left-hand side is minimized at a distribution $p$ which places all its probability mass on a single, worst case input state $\rho_i$.
Hence, this tells us that
\begin{equation*}
    \max_M \min_{i \in [\ell]} \{P_{\mathrm{succ}}(\rho_i, M)\} = \min_p \max_M \{P_{\mathrm{succ}}(p, M)\}.
\end{equation*}
The left-hand side is just the probability of success of the optimal measurement for shadow tomography in the traditional, worst-case setting, and the right-hand side is just the optimal probability of success in the average-case setting with prior $p$, minimized over all prior $p$.
Hence, just as in quantum state discrimination, if we can solve shadow tomography in the average-case setting, we can solve it in the worst-case setting.

\subsection{The PGM in shadow tomography}

The basic algorithm we will consider in this work is the natural plug-in estimator associated with the pretty good measurement.
Given $\brho^{\otimes n}$, where $\brho$ is drawn from the prior $p$, it acts as follows:
\begin{enumerate}
    \item Measure with the PGM associated with the prior $p$ on the states $\rho_1^{\otimes n}, \ldots, \rho_{\ell}^{\otimes n}$. Let $\bi \in [\ell]$ be the measurement outcome, and set $\widehat{\brho} = \rho_{\bi}$.
    \item Output $\widehat{\bo}_1 = \tr(O_1 \cdot \widehat{\brho})$, \ldots, $\widehat{\bo}_m = \tr(O_m \cdot \widehat{\brho})$.
\end{enumerate}
Surprisingly, this natural strategy has not previously been considered in the literature.
All previous attempts which have achieved sample complexity which is polylogarithmic in $m$ (as well as in $d$) have operated in the general framework proposed by Aaronson in his very first shadow tomography paper in~\cite{Aar18}, in which one maintains a hypothesis state $\widehat{\brho}$ which serves as one's guess for the true state, and one repeatedly updates $\widehat{\brho}$ by finding observables whose values it predicts incorrectly.
Most of the algorithms considered in these works apply to the online setting, in which one has to predict the value of the first observable before being given the second, and so on, and it is known that any algorithm in the online setting whose sample complexity is polylogarithmic in $m$ must also have at least a $\sqrt{\log(d)}$ dependence on the dimension.

A subset of our authors were originally convinced that the PGM strategy above \emph{must} be the right strategy for achieving sample complexity $O(\log(m))$, as it is mathematically elegant, and Math is typically fair.
Furthermore, presuming that it \emph{does} work, there is a clear path towards proving it,
which we formalize in the following conjecture.

\begin{conjecture}[PGM Decay Conjecture]
For any observable $O \in \C^{d \times d}$ with $\Vert O \Vert_{\infty} \leq 1$,
    \begin{equation*}
        \Pr\big[|\tr(O \cdot \widehat{\brho}) - \tr(O \cdot \brho)| > \eps\big] \leq O\big(C^{-n \epsilon^2}\big),
    \end{equation*}
    for some absolute constant $C >0$.
    Here, $\brho$ and $\widehat{\brho}$ are distributed as above.
\end{conjecture}

\noindent
If this conjecture were true, then applying it to the observables $O_1, \ldots, O_m$, we would have that, by a union bound,
\begin{equation*}
    \Pr\big[\text{exists $k \in [m]$ such that }|\widehat{\bo}_k - \bo_k| > \eps\big] \leq O(m \cdot C^{-n \epsilon^2}),
\end{equation*}
where we recall that $\bo_k = \tr(O_k \cdot \brho)$, for each $k \in [m]$ (and similarly for $\widehat{\bo}_k$).
Then if we set $n = O(\log(m)/\eps^2)$, the PGM will succeed with probability at least .99, as desired.
Note that as we have observed in \Cref{eq:pgm-fail} above, the PGM's failure probability \emph{does} decay exponentially as a function of the number of copies $n$ in the setting of quantum state discrimination, and so it is reasonable to hope that a similar behavior should occur here as well.
Unfortunately, the PGM Decay Conjecture actually turns out to be \emph{false}, and there are examples of families $\calF$, priors $p$, and observables $O$ in which the error probability of the PGM decays as $1/\poly(n)$ rather than our desired $1/\exp(n)$. In the next few subsections, we will discuss natural attempts at proving the PGM Decay Conjecture and why these fail, as these failures will illuminate features that our eventual algorithm will have.

Let us note one interesting feature of the above algorithm: the measurement does not depend on the observables we are trying to estimate;
these are only taken into account when the measurement outcome is classically post-processed.
In the traditional, worst-case setting, an algorithm for performing shadow tomography whose measurement does not depend on the observables is said to be solving the \emph{classical shadows} problem (named after the algorithm in the work~\cite{HKP20}), and it is known that classical shadows algorithms cannot succeed with $n = O(\log(m)/\eps^2)$ copies. Indeed, the work of Grier, Pashayan, and Schaffer~\cite{GPS24} shows that classical shadows algorithms require $\widetilde{\Omega}(\sqrt{d})$ copies, even to estimate a single observable.
However, these bounds only apply to the worst-case setting, and it is consistent with our state of knowledge that a $O(\log(m)/\eps^2)$-copy classical shadows algorithm could exist in the average case setting (indeed, we conjecture that this is the case).
As we have shown, this would imply an $O(\log(m)/\eps^2)$-copy algorithm for shadow tomography in the worst case. How could this be, given that we have worst-case lower bounds against classical shadows algorithms?
The answer is that the worst-case algorithm this reduction would produce is equal to the average-case algorithm for the prior $p$ which is hardest for the observables $O_1, \ldots, O_m$, and this prior would presumably encode the observables $O_1, \ldots, O_m$.

\subsection{A mean-squared error bound for the PGM} \label{subsec:pgm-mse}

The first difficulty one encounters when trying to prove the PGM Decay Conjecture is that very little work has been done on the PGM, outside of the literature on quantum state discrimination.
Results from this literature will help us understand with what probability the PGM will \emph{perfectly} identify the input state (or, perhaps, some feature of the input state),
but in our setting we only want the PGM to \emph{approximately} identify the input state, in a very weak sense (i.e.\ it should find a state which is close to the input on the observable $O$).
As a result, existing tools do not seem to apply to our setting.
The one exception we are aware of is the work of Mishra, Lami, and Wilde~\cite{MLW26},
who proved bounds on the ability of the PGM to approximately predict some feature of the input state.

\subsubsection{The Mishra--Lami--Wilde bound}
\label{sec:mishra-lami-wilde}

To state their result, suppose we have a set of states $\{\sigma_x\}_{x \in \calL}$, where the label set $\calL$ is a subset of $[-1, 1]$, and let $q = (q_x)_{x \in \calL}$ be a prior on this set.
We think of the label $x \in \calL$ as encoding some parameter of the state $\sigma_x$ that we would like to predict;
given a state $\sigma_{\bx}$ drawn from the prior $p$, the goal of an algorithm is to output a value $\widehat{\bx} \in \R$ such that $\widehat{\bx}$ is close to $\bx$.
(This is in contrast with the traditional setting of quantum state discrimination, where we only care about identifying $\bx$ exactly.)
The \emph{mean-squared error (MSE)} of the algorithm is given by
\begin{equation*}
    \E_{\bx \sim p} \big[\E_{\widehat{\bx}} (\widehat{\bx} - \bx)^2\big].
\end{equation*}
Let us write $\mathrm{MSE}_{\mathrm{OPT}}$ for the minimal MSE achievable by an algorithm on the prior $p$, and let us write $\mathrm{MSE}_{\mathrm{PGM}}$ for the MSE that the PGM achieves.
Then the main result of~\cite{MLW26} is that
\begin{equation}\label{eq:mlw}
\mathrm{MSE}_{\mathrm{PGM}} \leq 2 \cdot \mathrm{MSE}_{\mathrm{OPT}}.
\end{equation}
This is the analogue of the Barnum and Knill~\cite{BK02} bound from \Cref{eq:pretty-good} in the setting of estimating a parameter of a state, and the proofs of these two statements are similar.

Let us now apply this bound to derive a consequence related to the PGM Decay Conjecture.
Let $O \in \C^{d \times d}$ be an observable with $\Vert O \Vert_{\infty} \leq 1$ that we wish to estimate, and let $\calF = \{\rho_1, \ldots, \rho_{\ell}\}$ be our family of states.
For simplicity, we will assume that the observable values $\tr(O \cdot \rho_1), \ldots,\tr(O \cdot \rho_{\ell})$ are distinct,
and we will define the label set to be $\calL = \{\tr(O \cdot \rho_i)\}_{i \in [\ell]}$.
In addition, we will set the states and prior as follows: for the observable value $o \in \calL$ corresponding to the state $\rho_i \in \calF$, we set $\sigma_o = \rho_i^{\otimes n}$ and $q_o = p_i$.
Now we can apply the main result of \cite{MLW26}, which says that \Cref{eq:mlw} holds in this setting.

To interpret what this means, let us first understand what the PGM looks like for this ensemble.
First, the average of the states is
\begin{equation*}
    S
    = \sum_{o \in \calL} q_{o} \cdot \sigma_{o}
    = \sum_{i=1}^{\ell} p_i \cdot \rho_i^{\otimes n}.
\end{equation*}
In addition, for each $o \in \calL$, the POVM element associated with the label $o$ is given by
\begin{equation*}
    M_{o} = S^{-1/2} q_{o} \sigma_{o} S^{-1/2}
    = S^{-1/2} p_i \rho_i^{\otimes n} S^{-1/2},
\end{equation*}
where $\rho_i$  is the state in $\calF$ with observable value $o$.
This is identical to the PGM that we get from the ensemble $\{\rho_i^{\otimes n}\}$ with prior $p$, except with the output label classically post-processed via the mapping $i \mapsto o_i$.
The fact that the PGM is well-behaved under classical post-processing was previously observed and applied in the prior work of~\cite{KT08,BJL23,LdW25} and the independent and concurrent work of~\cite{AG26}, who referred to it as the PGM being ``closed under coarse-graining''.

Next, let us understand the optimal strategy (in terms of the mean-squared error) for estimating the observable  $O$.
One strategy for doing this, though perhaps not precisely \emph{the} optimal one, is to simply measure the observable $O$ on each copy of the input state independently and then average the results. 
To analyze this, let us first condition on the input state $\rho \in \calF$, and write $o$ for its observable value.
Write $\widehat{\bx}^1, \ldots, \widehat{\bx}^n$ for the results of the $n$ measurements, and let $\widehat{\bx}$ be their average.
Then each $\widehat{\bx}^i$ is an unbiased estimator for the true observable value, meaning $\E[\widehat{\bx}^i] = o$, and each $\widehat{\bx}^i$ falls within the interval $[-1, 1]$ (since $\Vert O \Vert_{\mathrm{op}} \leq 1$).
This means that $\Var[\widehat{\bx}^i]\leq 1$ for all $i$, and so $\Var[\widehat{\bx}] \leq 1/n$, since $\widehat{\bx}$ is their average.
As a result, we can bound the MSE of this algorithm by
\begin{equation*}
    \E_{\brho \sim p} \big[\E_{\widehat{\bx}} (\widehat{\bx} - \bo)^2\big]
    =  \E_{\brho \sim p} \big[\Var[\widehat{\bx} \mid \brho]\big]
    \leq \E_{\brho \sim p}[1/n] = 1/n.
\end{equation*}
This then implies that $\mathrm{MSE}_{\mathrm{OPT}} \leq 1/n$ as well, since this is the MSE of the algorithm with minimal MSE.
Hence, by \Cref{eq:mlw}, we have
\begin{equation}\label{eq:mse-for-one-observable}
\E_{\brho \sim p} \big[\E_{\widehat{\bo}} (\widehat{\bo} - \bo)^2\big]
=
\mathrm{MSE}_{\mathrm{PGM}} \leq 2 \cdot \mathrm{MSE}_{\mathrm{OPT}} \leq 2/n.
\end{equation}
Then a simple application of Markov's inequality tells us that $|\widehat{\bo} - \bo| \leq \epsilon$
except with probability $2/(n \epsilon^2)$.
This is similar to our desired bound in the PGM Decay Conjecture, though much weaker than it: we are hoping for a bound which is inverse exponential in $n \epsilon^2$, whereas this is only inverse polynomial.

For later, let us briefly note that the strategy from above satisfies a stronger form of concentration than just a mean-squared error bound.
In particular, since $\widehat{\bx}$ is an average of $n$ independent random variables contained in $[-1, 1]$, a Hoeffding bound will tell us that $|\widehat{\bx} - \bo| \leq \epsilon$ except with probability $2 e^{-n\epsilon^2/2}$.

\subsubsection{A new shadow tomography guarantee}

Even though this bound doesn't achieve the PGM Decay Conjecture, it still allows us to say something new and interesting about shadow tomography. 
Let $O_1, \ldots, O_m$ be the $m$ observables we would like to estimate,
and let $\widehat{\bo}_1, \ldots, \widehat{\bo}_m$ be the predictions made by the PGM-based shadow tomography algorithm.
Then applying \Cref{eq:mse-for-one-observable} to each observable separately, we have that
\begin{equation}\label{eq:mse-by-linearity}
    \E_{\brho\sim p} \E_{\widehat{\bo}} \Vert \widehat{\bo} - \bo \Vert_2^2
    = \E_{\brho \sim p} \E_{\widehat{\bo}}\sum_{i=1}^m (\widehat{\bo}_i - \bo_i)^2
    = \sum_{i=1}^m \E_{\brho \sim p}  \E_{\widehat{\bo}}(\widehat{\bo}_i - \bo_i)^2
    \leq \sum_{i=1}^m 2/n = 2m/n.
\end{equation}
We will try to interpret what this bound means shortly.
First, though, let us observe that although this is an average-case guarantee, one can use a similar minimax argument as before to show that it implies a worst-case guarantee as well.
In particular, doing so implies the following statement.

\begin{proposition}[Mean-squared error bound for shadow tomography]\label{prop:mse}
    For any observables $O_1, \ldots, O_m \in \C^{d \times d}$ with $\Vert O_i \Vert \leq 1$,
    there exists an algorithm such that given $n$ copies of an unknown (worst-case) state $\rho \in \C^{d \times d}$, the algorithm outputs estimates $\widehat{\bo}_1, \ldots, \widehat{\bo}_m \in [-1, 1]$ such that
    \begin{equation*}
        \E_{\widehat{\bo}} \Vert \widehat{\bo} - o \Vert_2^2 \leq \frac{2m}{n},
    \end{equation*}
    where $o_i = \tr(O_i \cdot \rho)$ for all $1 \leq i \leq m$.
\end{proposition}

To understand this bound, note that Markov's inequality implies that
\begin{equation}\label{eq:just-markoved}
    \Vert \widehat{\bo} - o \Vert_2^2 \leq O(m/n)
\end{equation}
with probability at least .99. If, for example, we set $n = O(1/\eps^2)$ to be large enough, then when this occurs, we have
\begin{equation*}
    \Vert \widehat{\bo} - o \Vert_2^2 \leq \frac{1}{100} m \eps^2.
\end{equation*}
An averaging argument then implies that for at least 99\% of the observables, we have $|\widehat{\bo}_i - o_i| \leq \epsilon$, whereas at most $1\%$ of them violate this.
More generally, given $O(\frac{1}{\alpha\eps^2})$ copies of $\rho$, we have an algorithm which, with probability $\geq 0.9$ (amplifiable through the median trick), accurately estimates a $1-\alpha$ fraction of the observables to within $\eps$ error.
This already is a non-trivial result: Aaronson considers this setting in \cite{Aar07}, and proposes a PAC learning approach to solving it.
His algorithm uses $\widetilde{O}(\frac{\log(d)}{\eps^2 \alpha})$ copies; we recover this complexity without the dimension dependence.
We note that a similar guarantee, with a similar high level proof strategy, had previously appeared in the literature on quantum random access codes (QRACs) due to a result of~\cite{LdW25}, which gives an algorithm for extracting most, but not all, of the bits of a QRAC.

Unfortunately, though, we do not necessarily know \emph{which} of the observables are incorrect; if we did, we could simply recurse on these ``bad'' observables in order to achieve a $\poly(\log(m),1/\epsilon)$ sample complexity.
Without knowing which are bad, then, the only obvious way to use this guarantee to estimate all $m$ observables is to apply \Cref{eq:just-markoved} with $n = O(m/\epsilon^2)$ copies, giving us
\begin{equation*}
    \Vert \widehat{\bo} - o \Vert_2^2 \leq \epsilon^2.
\end{equation*}
For this to be true, it must be that $|\widehat{\bo}_i - o_i| \leq \epsilon$ for all $1 \leq i \leq m$.
However, even trivial algorithms for shadow tomography can achieve $n = O(m/\epsilon^2)$ copies, making this sample complexity not so interesting.

All that said, the bound in \Cref{prop:mse} is still quite interesting, because we do not know any way to achieve this sort of mean-squared error guarantee without going through the PGM and the minimax theorem.
Early in the project, we viewed this as an encouraging sign that the PGM was the ``right'' approach to shadow tomography.
Not only that, the bound in \Cref{prop:mse} is the optimal bound that one can show on the mean-squared error of a shadow tomography algorithm.
To see this, we first note that this statement for general $m$ reduces to the case when $m = 1$.
This is because if we have an example of a prior and an observable $O$ such that any algorithm must satisfy $\E(\widehat{\bo}- \bo)^2 \geq \Omega(1/n)$, then we could simply take our $m$ observables to be $O_1 = \cdots = O_m = O$,
and then any algorithm must satisfy $\E \Vert \widehat{\bo} - \bo \Vert_2^2 \geq \Omega(m/n)$ by linearity.
Then, in the $m = 1$ case, an example of a prior and observable for which any algorithm must satisfy $\E(\widehat{\bo} - \bo)^2 \geq \Omega(1/n)$ is given by
\begin{align*}
    p_1 = p_2 = \frac{1}{2},\quad
    \rho_1 &= \Big(\frac{1}{2} + \frac{1}{\sqrt{n}}\Big) \cdot\ketbra{0} + \Big(\frac{1}{2} - \frac{1}{\sqrt{n}}\Big) \cdot\ketbra{1},\\
    \rho_2 &= \Big(\frac{1}{2} - \frac{1}{\sqrt{n}}\Big) \cdot\ketbra{0} + \Big(\frac{1}{2} + \frac{1}{\sqrt{n}}\Big) \cdot\ketbra{1}, \quad
    O = \ketbra{0}.
\end{align*}
Although we will not prove that the bound holds in this case, intuitively it is because $n$ samples are needed to estimate the bias of a coin up to error $1/\sqrt{n}$.

\subsection{Negative results for the PGM}

Following our success with the mean-squared error bound, we made various attempt at deriving $\poly\log(m)$-copy shadow tomography algorithms using the PGM, either by establishing the PGM Decay Conjecture or by other means.
In this subsection, we give an overview of some of these attempts, explain why they fail, and describe what they say about what an eventual PGM-based shadow tomography algorithm must look like.

\subsubsection{Bootstrapping the mean-squared error bound?}\label{sec:bootstrap}

Given nonnegative weights $w_1, \ldots, w_m$, consider the weighted mean-squared error given by
\begin{equation*}
    \E_{\brho \sim p} \sum_{i=1}^m w_i \cdot (\widehat{\bo}_i - \bo_i)^2.
\end{equation*}
If $w_1 + \cdots + w_m = m$, then the same argument that established \Cref{eq:mse-by-linearity} also implies that the weighted MSE of the PGM is at most $2m/n$.
Based on this, one strategy we tried was to ``bootstrap'' the seemingly-weak weighted mean-squared error bounds into an algorithm with $\polylog(m)$ copy complexity by repeatedly performing the PGM and adjusting the weights we use when calculating its weighted mean-squared error.
For example, suppose we run the PGM on our state with $n = O(1/\epsilon^2)$ copies.
As we have seen, it will correctly estimate 99\% of the observables up to error $\epsilon$, with reasonably high probability.
Now, suppose there \emph{was} a way, using our knowledge of the prior distribution $p$ and the values that the PGM output, to infer which of these predictions we have higher confidence in, and which of these predictions we have lower confidence in.
Then we could perform the PGM again, and this time when we analyze its mean-squared error, we could place more weight on those observables whose predictions we are less confident about, forcing the output of the second PGM to be more accurate on these observables than the first PGM was.
We could then repeat this until we are confident about all the observables.

\paragraph{A Pauli counterexample.}
Unfortunately, there is a simple counterexample which shows that this approach, at least as stated here, cannot work.
In this counterexample, let us suppose that $d$ is a power of $2$, and let our observables $O_1, \ldots, O_m$ be $m$ pairwise anticommuting Pauli matrices.
Let the family of states be $\calF=\{\rho_1, \ldots, \rho_m\}$, where for each $1 \leq i \leq m$,
\begin{equation*}
\rho_i = \frac{I + O_i}{d},
\end{equation*}
and let the prior be $p = (\tfrac1m,\ldots,\tfrac1m)$.
Then
\begin{equation*}
    \tr(O_j \cdot \rho_i) =
    \left\{\begin{array}{rl}
            1 & \text{if $i = j$,}\\
            0 & \text{otherwise}
            \end{array}\right.
\end{equation*}
Thus, performing shadow tomography on these $m$ observables with the error parameter set to a small constant, say $\eps = 0.1$, is equivalent to identifying which state one is given.
Now, suppose that after the first step of our bootstrapping scheme, the PGM returns $\widehat{\bo}_1, \ldots, \widehat{\bo}_m$ satisfying
\begin{equation}\label{eq:step-one-of-the-bootstrap}
    \Vert \widehat{\bo} - \bo \Vert_2^2 \leq O(m/n) = O(m\epsilon^2) = (\text{a very small constant}) \cdot m.
\end{equation}
To achieve this bound, it could be that the PGM simply outputs $\widehat{\bo}_1 = \cdots = \widehat{\bo}_m = 0$, which would satisfy the even stronger bound
\begin{equation*}
    \Vert \widehat{\bo} - \bo \Vert_2^2 = 1.
\end{equation*}
This is a possible output if the only property we use of the first PGM is that its outputs satisfy \Cref{eq:step-one-of-the-bootstrap}. But this tells us nothing interesting, as we could have output these values even without looking at the input state, and so we cannot use this to bootstrap any further.

Incidentally, another reason that bootstrapping fails on this example, at least for the particular form of bootstrapping that we consider here, is that if one only performs PGM measurements, then solving this example requires performing entangled measurements across $n = \Omega(\log(m))$ copies at a time (so, $O(1/\eps^2)$ copies at a time does not suffice).
This is because it can be shown that when the PGM is performed on this example with $t = o(\log(m))$ copies, the result will essentially be a uniformly random label in $[m]$, except with a slight bias towards the true label $i$,
in which case the PGM would have to be repeated many times before one could extract the true label.

\subsubsection{A high probability version of the Barnum--Knill bound?}
\label{sec:general-quantum-case}

In the previous section, we saw that the PGM inherits the mean-squared error guarantee of the optimal algorithm, with a factor-2 loss.
What if the optimal algorithm exhibits stronger forms of concentration? Will the PGM inherit these too?

In the setting of \Cref{sec:mishra-lami-wilde}, in which we have a prior $(q_x)_{x \in \calL}$ on states $\{\sigma_x\}_{x \in \calL}$, let us consider the probability that the algorithm succeeds at estimating the parameter $\bx$ of the input state $\sigma_{\bx}$, meaning that its output satisfies $|\widehat{\bx} - \bx| \leq \epsilon$.
Let us write $P_{\mathrm{OPT}}(q, \epsilon)$ for the success probability of the optimal measurement (and $\overline{P}_{\mathrm{OPT}}(q, \epsilon) = 1 - P_{\mathrm{OPT}}(q, \epsilon)$ for its failure probability), and define $P_{\mathrm{PGM}}$ and $\overline{P}_{\mathrm{PGM}}$ similarly. Could it be that, say
\begin{equation}\label{eq:conjectured-bk-analogue}
    P_{\mathrm{PGM}}(q, \epsilon) \geq P_{\mathrm{OPT}}(q, \epsilon)^2,
    \quad \text{and therefore} \quad
    \overline{P}_{\mathrm{PGM}}(q,\epsilon) \leq 2 \cdot \overline{P}_{\mathrm{OPT}}(q, \epsilon)?
\end{equation}
This would be a natural analogue of the Barnum--Knill bound in \Cref{eq:pretty-good}, except to the problem of approximately computing a parameter, rather than exactly computing it.
And a bound of this form would be enough to prove our PGM Decay Conjecture.
This is because as noted in \Cref{sec:mishra-lami-wilde} above, there is an algorithm whose probability of failure can be bounded by $2e^{-n\epsilon^2/2}$, and so the probability of failure of the optimal algorithm can as well. 
Therefore, \Cref{eq:conjectured-bk-analogue} would imply that $\overline{P}_{\mathrm{PGM}} \leq 4 e^{-n\epsilon^2/2}$, as claimed by the PGM Decay Conjecture.

\ignore{
\subsubsection{The classical special case}

How reasonable is \Cref{eq:conjectured-bk-analogue}?
As it turns out, it \emph{does} hold (or at least a small variant of it does) in the classical special case, which is the case when all the states $\sigma_x$ are diagonal in the standard basis.
In this case, it is known (see, for example, XXX) that the problem is equivalent to a problem involving probability distributions.
In the distribution version, for each $x \in \calL$, there is a discrete probability distribution $a_x = (a_{x,1}, \ldots, a_{x, d})$ on the set $[d]$.
We draw $\bx \sim q$, and then the algorithm is provided a sample $\bi \sim a_x$.
Given this, the algorithm outputs an estimate $\widehat{\bx}$, and the goal is for $|\widehat{\bx} - \bx| \leq \epsilon$.
Given $\bi \in [d]$, an algorithm can reason about the parameter it is trying to estimate by inspecting the Bayesian posterior distribution $\Pr[x \mid \bi]$, which will tell it the likelihood that the sample $\bi$ came from a given distribution $a_x$.
In this setting, then, the PGM is the algorithm which performs \emph{posterior sampling}, i.e.\ sampling its output $\widehat{\bx}$ from the posterior distribution.

To compare the PGM to the optimal algorithm, suppose that on input $\bi \in [d]$, the optimal algorithm deterministically outputs a value $\widehat{y}_{\bi}$. (That the optimal algorithm is deterministic in its input is without loss of generality.)
Let us write $s_{\bi}$ for the success probability of the optimal algorithm, condition on receiving input $\bi \in [d]$.
This means that if we consider the interval $I_{\bi} = [\widehat{y}_{\bi} - \epsilon, \widehat{y}_{\bi} + \epsilon]$, then $\bx$ falls within this interval with probability $s_{\bi}$; in other words, the posterior distribution $\Pr[\bx \mid \bi]$ places a total probability of $s_{\bi}$ on points $\bx$ within this interval.
Since the PGM samples from the posterior distibution, this means that the PGM's output $\widehat{\bx}$ also falls within the interval $I_{\bi}$ with probability equal to $s_{\bi}$.
So, conditioned on the value of $\bi$, the probability that \emph{both} $\bx$ and $\widehat{\bx}$ fall within $I_{\bi}$ is $s_{\bi}^2$. And when this happens, we always have that $|\widehat{\bx} - \bx|\leq2 \epsilon$. 
Averaging over all $\bi \sim p$, this tells us that the probability that the PGM's output is within $2\epsilon$ of the true parameter $\bx$ is at least
\begin{equation*}
    P_{\mathrm{PGM}}(q,2\epsilon) \geq \E_{\bi}[s_{\bi}]^2 \geq (\E_{\bi} [s_{\bi}])^2 = P_{\mathrm{OPT}}(q,\epsilon)^2,
\end{equation*}
where in the first step we used Jensen's inequality.
This shows that a variant of \Cref{eq:conjectured-bk-analogue} is true in the classical case, with a proof that is similar to the Barnum--Knill proof of \Cref{eq:pretty-good}.

\subsubsection{The general quantum case}
}

\paragraph{A quantum counterexample.}
As it turns out, \Cref{eq:conjectured-bk-analogue} is true in the classical case, when the states are all diagonal in the standard basis, via a proof that is similar to the Barnum--Knill proof of \Cref{eq:pretty-good}.
We attempted to generalize this argument to the quantum case, but we found doing so challenging.
Eventually, we discovered why: there is a strong counterexample, in which the optimal success probability is exactly 1, but the PGM does significantly worse.
Suppose we are in $d$-dimensional space, with $d$ odd, and have chosen to index the standard basis as $\ket{0}, \ldots, \ket{d-1}$.
Consider the family of pure states given by
\begin{equation*}
    \ket{v_i} = \frac{1}{\sqrt{2}} \cdot \ket{i~\text{(mod d)}} + \frac{1}{\sqrt{2}} \cdot \ket{i + 1 ~\text{(mod d)}},
    \qquad \text{for $i \in \{0, \ldots, d-1\}$}.
\end{equation*}
We will assume the uniform prior on this set of states.
Given a state $\ket{v_i}$ from this family, the ``parameter'' we want to estimate of this state is given by the function
\begin{equation*}
    f(i) =
    \left\{\begin{array}{ll}
    \frac{2}{d-1} \cdot i, & \text{if }i \in \{0, \ldots, (d-1)/2\},\\
    \frac{2}{d-1} \cdot (d-i), & \text{if }i \in \{(d-1)/2 + 1, \ldots, d-1\}.
    \end{array}\right.
\end{equation*}
This function increases linearly from 0 to 1 as $i$ goes from $0$ to $(d-1)/2$, and then it decreases from there linearly back to 0.
We have chosen to have this particular function so that the parameter of $\ket{v_i}$ changes slowly as we increase $i$, and the two different regimes are chosen so that the pure state
\begin{equation*}
    \ket{v_{d-1}} = \frac{1}{\sqrt{2}} \ket{d-1 ~\text{(mod 2)}} + \frac{1}{\sqrt{2}}\ket{d ~\text{(mod 2)}}
    = \frac{1}{\sqrt{2}} \ket{d-1} + \frac{1}{\sqrt{2}}\ket{0}
\end{equation*}
has a parameter which is close to that of its two neighbors $\ket{v_d}$ and $\ket{v_0}$.
Note that here we have decided to separate the way we index into the states from the parameters of the states we wish to estimate; this is done for convenience.

Given a state $\ket{v_{\bi}}$ sampled from this prior, suppose we want to estimate its parameter up to error $\eps = 2/(d-1)$.
Then there is an algorithm for this task which succeeds with probability 1: simply measure the state in the standard basis, let $\bj \in \{0, \ldots, d-1\}$ be the outcome, and output $f(\bj)$.
This will always either output the correct parameter value, or a value which is within $\eps = 2/(d-1)$ of it,
as each state only has support on $i$ and $i+1$ (modulo $d$), and so this very simple algorithm will always succeed.

But how does the PGM do? It turns out that for a parameter $\epsilon'$, its output will be farther from the true value by at least $\epsilon'$ with probability that scales like $1/\poly(d \epsilon')$.
So for $\epsilon'$ any constant multiple of $\epsilon=2/(d-1)$, its output will be off by at least $\epsilon'$ with probability $\Omega(1)$, which rules out \Cref{eq:conjectured-bk-analogue} as well as versions of it which are weaker by constant factors.
We will not prove this, though it is not  a difficult argument to carry out, but let us sketch out a few of its parts for intuition's sake.
The PGM has measurement operators $\{M_i\}_{i \in \{0, \ldots, d-1\}}$, which are defined via
\begin{align*}
    M_i = \ketbra{w_i}, \quad \text{where } \ket{w_i} = S^{-1/2} \cdot \ket{v_i} \text{ and } S = \sum_{i=0}^{d-1} \frac{1}{d} \cdot \ketbra{v_i}.
\end{align*}
Computing $S^{- 1/2}$ is fairly easy, because $S$ is a circulant matrix, and it is well-known that the eigenvectors of a circulant matrix are given by the Fourier vectors
\begin{equation*}
    \ket{\psi_i} = \frac{1}{\sqrt{d}} \sum_{j=0}^{d-1} \omega^{i j} \cdot \ket{j},
\end{equation*}
where $\omega$ is the $d$-th root of unity.
The vectors $\ket{\psi_i}$ with $i$ small are ``local'' in the sense that they put very similar coefficients on standard basis vectors $\ket{j}$ which have similar values of $j$, whereas when $i$ is large, these vectors are ``nonlocal'', and may put very different coefficients on two nearby basis vectors.
The vectors $\ket{v_i}$ are similarly ``local'', and put most of their amplitude on the Fourier vectors $\ket{\psi_i}$ with small $i$,
and the matrix $S$ also has most of its eigenweight placed on these ``small $i$'' Fourier vectors.
Inverting this, $S^{-1/2}$ will therefore place a \emph{large} amount of weight on these ``large $i$'' Fourier vectors, and so multiplying by it will serve to amplify the weight on the ``large $i$'' Fourier vectors. This means that the $\ket{w_i} = S^{-1/2} \cdot \ket{v_i}$ vectors will have their mass much more spread out than the $\ket{v_i}$ vectors in the standard basis, and so measuring a vector $\ket{v_i}$ with the PGM will often result in a faraway measurement outcome $M_j$.

\subsubsection{Making use of the tensor product structure?}
\label{sec:tensor}

The counterexample in the previous section shows that we probably cannot hope to prove the PGM Decay Conjecture by proving a statement about the PGM's success probability on a generic, unstructured ensemble of states, and then use it to derive the PGM Decay Conjecture as a special case (in the same way that we could derive \Cref{eq:mse-for-one-observable} as a special case of \Cref{eq:mlw}).
Instead, we would have to use specific properties about our particular ensemble, namely the fact that our ensemble consists of $n$-fold tensor products of $d$-dimensional states (whereas the counterexample in \Cref{sec:general-quantum-case} is not a tensor product). Could this be the missing piece?

\paragraph{A single qubit counterexample.}
We were able to generate numerical evidence suggesting that the answer was still no, and that in fact the PGM sometimes has subexponential decay, perhaps as bad as $e^{-\sqrt{n}}$.
This evidence was based on an ensemble of $\ell$ qubit pure states $\ket{v_0}, \ldots, \ket{v_{\ell}}$ of the form
\begin{equation*}
\ket{v_i}=
\cos(\theta_i) \cdot \ket{0} + \sin(\theta_i) \cdot \ket{1},
\end{equation*}
with the $\ell$ angles, very roughly, being evenly spaced between $0$ and $\pi/2$ (so that $\theta_i \approx i \cdot \pi/(2(\ell-1))$).
We considered various ways to modify the PGM to thwart this ensemble, without success.
These included adding noise to the ensemble, with the hope of combating the issue of $S$ being poorly conditioned, as well as considering other variants of the PGM. To expand on this, the PGM is defined by computing the $S$ matrix, and then setting the measurement operator $M_i$ equal to $q_x \sigma_x$ ``divided by'' $S$.
There is, however, no unique way to define what it means to divide one matrix $A$ by another matrix $B$.
The traditional definition of the PGM takes this to mean $B^{-1/2} A B^{-1/2}$, but there is another natural way to divide matrices based on the Sylvester superoperator,
and yet another based on the derivative of the matrix logarithm. This last one corresponds to what is known as the ``integral PGM'', which was first considered in the work~\cite{CL26a} and has more recently been applied in~\cite{CL26b}.
However, none of these seemed to fare any better on this family of counterexamples.

\subsection{Our use of artificial intelligence} \label{subsec:ai}

Eventually, we were able to use artificial intelligence to make significant progress on this problem.
First, we were able to use GPT 5.5 to give a counterexample to the PGM Decay Conjecture, one which shows that the PGM's error probability can decay as slowly as $1/\poly(n)$.
(In fact, current models seem to be able to produce many such counterexamples.)
Next, using GPT 5.6, we were able to produce a proof showing that shadow tomography can be solved with $O(\log^3(m)/\eps^2)$ copies.
This proof uses a good number of the ideas here---the minimax theorem, the PGM, the result of~\cite{MLW26}, etc.---but frustratingly, the A.I.\ figured out these ideas on its own, without any human input.
This GPT proof is based on a more sophisticated version of the bootstrapping idea from \Cref{sec:bootstrap}.
Eventually, we were able to get GPT to produce a $O(\log^2(m)/\epsilon^2)$ copy shadow tomography algorithm using a substantially different proof (again, this proof was discovered with very little human input).
After some amount of refinement that we participated in, this is the proof that we will present here.

\paragraph{An exponential decay for the PGM.}
This proof addresses a basic mystery about the behavior of the PGM that our investigations raised: our intuition from the state discrimination literature, and in particular from \Cref{eq:pgm-fail} above, is that the PGM should improve inverse-exponentially as the number of copies of the state increases.
However, the counterexample to the PGM Decay Conjecture shows that there are examples where the PGM's failure probability only decays as inverse-polynomial.
So if it isn't the failure probability that decays as an inverse-exponential, what is it?
Our proof provides an answer in terms of the ``Good'' and ``Bad'' parts of the prior, which are defined as follows:
\begin{equation*}
    \mathsf{Good}_O \coloneqq \{\rho \in \calF \mid \Pr_{\widehat{\bo}}[|\widehat{\bo} - o| \leq \epsilon \mid \rho] \geq 7/8\},
    \qquad
    \mathsf{Bad}_O = \calF \setminus \mathsf{Good}_O.
\end{equation*}
In other words, conditioned on the state being $\rho$, what is the probability that the PGM's output is within $\epsilon$ of its observable value $o = \tr(O \cdot \rho)$? If the answer is at least $7/8$, then $\rho$ is Good, and if the answer is less than $7/8$, then the answer is Bad.
Then our main result is that the probability that a state from the prior falls into the Bad set decays exponentially in the number of copies $n$:
\begin{equation}\label{eq:no-bad}
    \Pr_{\brho \sim p}[\brho \in \mathsf{Bad}_O] \leq O(C^{-n \epsilon^2}).
\end{equation}
With this in place, union bounding over our $m$ observables $O_1, \ldots, O_m$, we have that
\begin{equation*}
    \Pr_{\brho \sim p}[\exists i \in [m] \mid \brho \in \mathsf{Bad}_{O_i}] \leq O(m \cdot C^{-n \epsilon^2}).
\end{equation*}
Hence by taking $n = O(\log(m)/\epsilon^2)$ copies, we can guarantee that with probability at least 0.99, the state $\brho$ from the prior will be Good for all $m$ observables. 
If this is the case, then we can simply repeat the PGM $\log(m)$ times and take the median value seen for each observable, and the algorithm will succeed with high probability, giving us our $O(\log(m)^2/\epsilon^2)$ sample complexity.

Now let us observe how this result is in line with the intuition we have gained so far.
First, the algorithm uses PGM measurements on $O(\log(m))$ copies at a time, which we saw was necessary in \Cref{sec:bootstrap}.
Second, the bound in \Cref{eq:no-bad} still allows the PGM to have error probability that decays as slowly as $1/\poly(n)$, as our counterexamples show is possible. (Indeed, even if \emph{every} state in $\calF$ is Good, this only guarantees an error probability of at most $1/8$!)
Finally, as we will eventually see, the proof of this statement inherently relies on the fact that the PGM is being applied to tensor product states, rather than applying to general ensembles, as \Cref{sec:general-quantum-case} suggested might be necessary.
w

We note two caveats about \Cref{eq:no-bad}: first, it is much more convenient for us to consider the integral PGM when proving it rather than the traditional PGM. (It \emph{is} possible to prove it with the traditional PGM as well, but doing so requires shoehorning the traditional PGM into a proof that is much more well-suited for the integral PGM.)
Second, the algorithm we actually study does not use a fixed number of copies $n$ for the integral PGM, but rather it samples the number of copies to be fed to the integral PGM uniformly at random from $n/2$ and $n$.
It is not clear to us if this is a necessary ingredient of the algorithm, or if it is just a technical limitation of the current proof strategy.

\section{Preliminaries}

For $u\in[0,1]$, we write $h(u) \coloneqq -u\log u-(1-u)\log(1-u)$ for the binary entropy function.
For a positive semidefinite operator $X$, let $\supp(X)$ denote the span of its eigenvectors with nonzero eigenvalues.
For quantum states $\rho$ and $\sigma$, we write 
\[
D(\rho \| \sigma) \coloneqq \tr\big(\rho (\log \rho - \log \sigma) \big)
\]
for their relative entropy, provided $\supp(\rho) \subseteq \supp(\sigma)$; otherwise we set $D(\rho \| \sigma) = \infty$. We denote the von Neumann entropy of $\rho$ by $H(\rho)$:
\begin{equation*}
    H(\rho) \coloneq -\tr\big( \rho \log \rho \big). 
\end{equation*}

\subsection{The integral PGM}


We begin by recalling a variant of the pretty good measurement that will be used in our proof.
Let $\{\rho_x\}_{x\in\calL}$ be a finite family of quantum states, let $p=(p_x)_{x\in\calL}$ be a prior, and define the \emph{average state} as 
\[
    S(p,n) \coloneqq \sum_{x\in\calL} p_x \rho_x^{\otimes n}.
\]
We will usually just write $S$ when the prior $p$ and number of copies $n$ are clear from context.

One way to think about the pretty good measurement is as follows.
For each $x$, we would like to construct a measurement operator by ``dividing''
$p_x\rho_x^{\otimes n}$ by the average state $S$.
More precisely, suppose that $\calD_S$ is a positive linear map which satisfies $\calD_S(S)=I$ on $\supp(S)$; then the operators
\[
    M_x \coloneqq \calD_S\!\left(p_x\rho_x^{\otimes n}\right)
\]
form a POVM on $\supp(S)$. This is because $M_x\succeq0$ for each $x$ by the positivity of $\calD_S$ and because
\[
    \sum_x M_x
    = \calD_S\!\Big(\sum_xp_x\rho_x^{\otimes n}\Big)
    = \calD_S(S)
    = I.
\]
Since every state $\rho_x^{\otimes n}$ is supported on $\supp(S)$, we may complete this measurement arbitrarily on the orthogonal complement.

The standard pretty good measurement corresponds to the choice
\[
    \calD_S(X) = S^{-1/2}XS^{-1/2},
\]
where inverses are taken on $\supp(S)$.
In this work, it will instead be more convenient to use a different notion of division by $S$.

\begin{definition}[The integral pretty good measurement]
\label{def:integral-pgm}
Let $\Pi$ denote the projector onto $\supp(S)$.
The \emph{integral pretty good measurement} associated with the ensemble
$\{p_x,\rho_x^{\otimes n}\}_{x\in\calL}$ is defined by $M_x
 \coloneqq \calD_S\!\left(p_x\rho_x^{\otimes n}\right)$, where
\[
    \calD_S(X)
    \coloneqq \int_0^\infty (S+t\Pi)^{-1} \cdot X \cdot (S+t\Pi)^{-1}\,\dt.
\]
\end{definition}

It is easy to verify that $\calD_S$ is positive and satisfies $\calD_S(S)=I$, and hence that the integral PGM is indeed a valid measurement; we omit the details.

The main reason we use the integral PGM is that $\calD_S$ is the derivative of the matrix logarithm at $S$:
\[
    \calD_S(X) = \left. \frac{\mathrm{d}}{\mathrm{d}u} \log(S+uX) \right|_{u=0}.
\]
This connection allows us to relate probabilities produced by the integral PGM to the matrix-logarithm expressions that appear in relative entropy. 
In particular, the following consequence of this identity will be the main property of the integral PGM that we use in our analysis.

\begin{lemma}[A useful inequality for the integral PGM] \label{lem:integral_PGM_inequality}
    Let $\sigma$ be a quantum state and let $X \succeq 0$ satisfy $\supp(\sigma) \subseteq \supp(X) \subseteq \supp(S)$. Then 
    \begin{equation*}
        \tr \big(
        \sigma \cdot ( \log X - \log S ) \big) \leq \log \tr \big( \sigma \cdot \calD_S(X) \big).
    \end{equation*}
\end{lemma}

\begin{proof}
The assumption $\supp(\sigma) \subseteq \supp(X) \subseteq \supp(S)$ allows us to restrict attention to $\supp(S)$. On this subspace, $S$ has no zero eigenvalues and is therefore positive definite.
We first verify that $\calD_S$ is the derivative of the matrix logarithm at $S$.
Recall the standard integral representation
\[
    \log A
    = \int_0^\infty \bigg( \frac{I}{1+t}-(A+tI)^{-1} \bigg)\, \dt, \qquad A\succ0,
\]
which follows by applying the corresponding scalar identity to the eigenvalues of $A$.
Thus, for any Hermitian matrix $H$ and $u$ sufficiently close to zero,
\[
    \log(S+uH)
    = \int_0^\infty \bigg( \frac{I}{1+t}-(S+uH+tI)^{-1}
  \bigg)\, \dt.
\]
By the Leibniz integral rule, we may differentiate under the integral sign, giving
\begin{align*}
    \frac{\mathrm{d}}{\mathrm{d}u}\log(S+uH)
    &=
    -\int_0^\infty
    \frac{\partial}{\partial u}
    (S+uH+tI)^{-1} \, \dt\\
    &=
    \int_0^\infty
    (S+uH+tI)^{-1} \cdot 
    H \cdot 
    (S+uH+tI)^{-1} \, \dt,
\end{align*}
where in the second line we used the standard identity
\[
    \frac{\mathrm{d}}{\mathrm{d}u}A(u)^{-1}
    =
    -A(u)^{-1} \cdot A'(u) \cdot A(u)^{-1}.
\]
Evaluating at $u=0$ gives
\begin{equation}
\label{eq:integral_PGM_matrix_log_derivative}
    \left.
    \frac{\mathrm{d}}{\mathrm{d}u}\log(S+uH)
    \right|_{u=0}
    =
    \calD_S(H).
\end{equation}

We now use \cref{eq:integral_PGM_matrix_log_derivative} to prove the desired inequality. Suppose first that $X$ is positive definite on $\supp(S)$, and set $a\coloneqq\tr\!\left(\sigma \cdot \calD_S(X)\right)$.
We will compare $S$ with the rescaled operator $X/a$. Define
\[
    f(u)
    \coloneqq
    \tr\Big(
        \sigma \cdot \log\big((1-u) \cdot S+u \cdot X/a\big)
    \Big),
    \qquad 0\leq u\leq1.
\]
By operator concavity of the logarithm, $f$ is a concave function of $u$.

The choice of the normalization $a$ ensures that $f'(0)=0$. Indeed, since $(1-u) \cdot S+u \cdot X/a = S+u \cdot (X/a-S)$, \cref{eq:integral_PGM_matrix_log_derivative} gives
\begin{align*}
    f'(0)=
    \tr\big(
        \sigma \cdot \calD_S(X/a-S)
    \big) =
    \frac{1}{a} \cdot \tr\big(\sigma \cdot \calD_S(X)\big)
    -
    \tr\big(\sigma \cdot \calD_S(S)\big) =1-1=0,
\end{align*}
where we have used the definition of $a$, together with
$\calD_S(S)=I$ and $\tr(\sigma)=1$.

Since $f$ is concave, it lies below its tangent line at $u=0$. Hence $f(1) \leq f(0)+f'(0) = f(0)$.
Substituting the definition of $f$ gives
\[
    \tr\!\big(\sigma \cdot \log(X/a)\big)
    \leq
    \tr(\sigma \cdot \log S).
\]
Finally, since $a$ is a scalar and $\tr(\sigma)=1$,
\[
    \tr\!\big(\sigma \cdot \log(X/a)\big) = \tr(\sigma \cdot \log X)-\log a.
\]
Rearranging therefore gives the desired inequality:
\[
    \tr\!\big(\sigma \cdot (\log X-\log S)\big)
    \leq
    \log a
    =
    \log\tr\!\big(\sigma \cdot \calD_S(X)\big).
\]

Finally, suppose that $X$ is singular. For $\eta>0$, the operator
$X+\eta I$ is positive definite, so the preceding argument applies to it.
We then let $\eta\to 0^+$.
The assumption $\supp(\sigma)\subseteq\supp(X)$ ensures that $\sigma$
has no support on the zero-eigenspace of $X$, and hence we have the convergence
\[
    \tr\!\big(\sigma \cdot \log(X+\eta I)\big)
    \to
    \tr(\sigma \cdot \log X).
\]
On the other hand, by linearity of $\calD_S$, $\tr(\sigma \cdot \calD_S(X+\eta I)) \to \tr(\sigma \cdot \calD_S(X))$.
Taking $\eta\to 0^+$ therefore gives the desired inequality for singular $X$ as well.
\end{proof}

\section{Average-case shadow tomography}

We execute the PGM-based strategy outlined in \cref{background}: we solve the shadow tomography problem in the average-case setting and then appeal to the minimax theorem to obtain an algorithm in the worst-case setting.
We begin by formally defining the average-case problem.

\begin{problem}[Average-case shadow tomography]
\label{prob:average-case-shadow-tomography}
Let $\calF=\{\rho_x\}_{x\in\calL}$ be a known finite family of $d$-dimensional quantum states, let $p=(p_x)_{x\in\calL}$ be a known prior, and let $O_1,\ldots,O_m$ be known observables satisfying $\norm{O_i}_\infty\leq1$.
An input state $\brho=\rho_{\bx}$ is drawn according to $\bx\sim p$. 
Given copies of $\brho$, the goal is to output estimates $\wh{\bo}_i$ of $\bo_i \coloneqq \tr(O_i \cdot \brho)$ for each $i \in [m]$ such that 
\[
    \Pr_{\brho,\wh{\bo}}\Big[
        \max_{i\in[m]} \Big\{\big| \wh{\bo}_i-\tr(O_i \cdot \brho) \big|\Big\} \leq\eps
    \Big]\geq 1 - \delta.
\]
The probability is over both the draw of the input state and
the randomness of the algorithm.
\end{problem}

The main result of this section is an algorithm for average-case shadow tomography. 

\begin{theorem}
\label{thm:average-case-shadow-tomography}
For any $\delta, \, \epsilon > 0$, there is an algorithm that solves \cref{prob:average-case-shadow-tomography} using
\[
    O\left(
        \frac{\log(m)\log(m/\delta)}{\eps^2}
    \right)
\]
copies of the input state.
\end{theorem}

The basic measurement primitive underlying \cref{thm:average-case-shadow-tomography} is a single application of the integral PGM, with the number of copies fed to the measurement chosen at random. 
Given a prior $p$ on a family of states $\calF$, the measurement outputs a candidate state, whose expectation values we use as our estimates.
We formalize this procedure below.

\begin{myalgorithm}[label={alg:integral-pgm-estimator}]
    {Integral PGM estimator}

Given a known family of states $\calF=\{\rho_x\}_{x\in\calL}$,
a prior $p=(p_x)_{x\in\calL}$, observables
$O_1,\ldots,O_m$, and integers $k,T\geq1$:

\begin{enumerate}
    \item Sample $\bt$ uniformly from $\{T,\ldots,2T-1\}$ and set $\bn \coloneqq k\bt$.

    \item Apply the integral PGM (\cref{def:integral-pgm}) associated with the ensemble $\{p_x,\rho_x^{\otimes \bn}\}_{x\in\calL}$ to $\bn$ copies of the input state.
    Let $\wh{\bx}\in\calL$ denote the outcome.

    \item Output $\wh{\bo}\in[-1,1]^m$ with $\wh{\bo}_i\coloneqq \tr(O_i\cdot\rho_{\wh{\bx}}) $ for each $i\in[m].$
\end{enumerate}
\end{myalgorithm}


The output of \cref{alg:integral-pgm-estimator} is a random variable that depends on the randomness of the input $\bx$ as well as the randomness generated by the algorithm itself, in sampling $\bt$ and measuring with the integral PGM. We will write $\wh{\bo} \sim \calA(x,n)$ when we want to describe the output distribution, conditioned on a certain input or number of copies, and drop these arguments when we do not condition on their values. For example, $\wh{\bo} \sim \calA$ conditions on neither input nor number of copies used. We will sometimes abuse notation and write $\widehat{\bx} \sim \calA$ as well.

The key property of the estimator described in \cref{alg:integral-pgm-estimator} is that, for any fixed observable, only an exponentially small fraction of the prior consists of states on which the estimator produces ``bad'' estimates. 
We formalize this in the following lemma, which is the main technical ingredient in the proof of \Cref{thm:average-case-shadow-tomography}.

\begin{lemma}
\label{lem:randomized-integral-pgm}
Let $p$ be any prior on $\calF$, and run
\cref{alg:integral-pgm-estimator} with $k=\left\lceil \tfrac{12}{\eps^2}\right\rceil$. Then, for every $T\geq1$ and every $i\in[m]$,
\[
    \Pr_{\bx\sim p}\bigg[
        \Pr_{\wh{\bo} \sim \calA(\bx)}\Big[ \big|\wh{\bo}_i-\tr(O_i \cdot \rho_{\bx})\big| >\eps 
        \Big]
        > \frac{1}{8}
    \bigg]
    \leq 2e^{-T/128}.
\]
\end{lemma}

We defer the proof of \cref{lem:randomized-integral-pgm}, and first show how it implies the average-case shadow tomography algorithm.
In what follows, we prove a slightly weaker $O(\log^2(m/\delta)/\eps^2)$ copy complexity. 
A somewhat more involved argument improves this to the bound stated in \cref{thm:average-case-shadow-tomography}, the details of which are provided in \cref{app:boost}.
We choose to present the weaker bound here because its proof is more transparent and already captures the main idea.

\begin{proof}[Proof of (the slightly weaker version of) \cref{thm:average-case-shadow-tomography}, assuming \cref{lem:randomized-integral-pgm}]
For any prior $q$ on $\calF$, run
\cref{alg:integral-pgm-estimator} with parameters $k = \lceil 12/\eps^2 \rceil$ and $T$ to be specified shortly, and let $e_q(x)$ denote the largest conditional error probability among the $m$ observable estimates on input $\rho_x$. That is,
\begin{equation*}
    e_q(x) \coloneqq \max_{i\in[m]} \bigg\{\Pr_{\wh{\bo} \sim \calA(x) }\Big[ \big|\wh{\bo}_i-\tr(O_i \cdot \rho_x)\big|>\eps 
    \Big] \bigg\}.
\end{equation*}
Here, the probability is over the randomness of \cref{alg:integral-pgm-estimator} when run with prior $q$.
By \cref{lem:randomized-integral-pgm} and a union bound over the $m$
observables,
\begin{equation*}
 \Pr_{\bx\sim q}\Big[e_q(\bx)>\frac{1}{8}\Big]
    \leq 2m e^{-T/128}.
\end{equation*}
Choosing $T = 128 \log (4m/\delta)$ suffices to make this probability at most $\delta/2$. 

Our goal now is to amplify this guarantee for those $x$ such that $e_q(x) \leq 1/8$; the remaining $x$ will be sufficiently rare as to not matter. We amplify using the median trick, a standard boosting argument. For an odd integer number of rounds $R$, run \cref{alg:integral-pgm-estimator} with the same parameters as above, obtaining the vectors $\widehat{\bo}^{(1)}, \dots, \widehat{\bo}^{(R)}$. Our final output is the coordinate-wise median of these vectors, i.e.\ $\widehat{\bo}_i \coloneq \mathrm{median}(\widehat{\bo}^{(1)}_i, \dots, \widehat{\bo}^{(R)}_i)$ for each $i \in [m]$. Fix an $i$, and define $\bF_{\!i}$ to be set of indices $j \in [R]$ which fail to be $\eps$-close to the true value:
\begin{equation*}
    \bF_{\!i} \coloneq \big\{ j \, : \, | \widehat{\bo}_i^{(j)} - o_i | > \eps \big\}.
\end{equation*}
Then the median of $\widehat{\bo}^{(1)}_i, \dots, \widehat{\bo}^{(R)}_i$ is more than $\eps$ away from $o_i$ only when $|\bF_{\!i}| > R/2$. 

Now, condition on $\bx = x$, such that $e_q(x) \leq 1/8$. Since $x$ is fixed, the indicator random variables $\{1[j \in \bF_i]\}_{j \in [R]}$ are independent and identically distributed, and 
\begin{equation*}
\E\big[ | \bF_{i}| \big] \leq e_q (x) \cdot R \leq \frac{R}{8}. 
\end{equation*}
Therefore, by an additive Chernoff bound,
\begin{equation*}
\begin{split}
    \Pr\Big[\big|\mathrm{median}(\widehat{\bo}^{(j)}_{i})_{j \in [R]} - \tr( O_i \rho) \big| > \epsilon \Big]  \leq \Pr\Big[ |\bF_{\!i}| > \frac{R}{2} \Big] \leq \Pr\Big[ |\bF_{\!i}| > \E\big[\bF_{\!i}\big] + \frac{3R}{8} \Big] \leq e^{- 2R \cdot (3/8)^2} =e^{-9R/32}.
\end{split}
\end{equation*}
Thus, taking $R = 4 \log(2m/\delta)$ suffices to make this probability less than $\delta/2m$. Then, by a union bound, the probability that any of the $m$ of the $\widehat{\bo}_i$ are not within $\eps$ of the correct value is at most $\delta/2$, as desired. 

Finally, relaxing the condition, with probability $2me^{-T/128}$ we draw a state which does \emph{not} satisfy $e_q(x) \leq 1/8$, so that
\begin{equation*}
    \Pr \Big[ \max_{i \in [m]} \big| \widehat{\bo}_i - o_i \big| > \eps \Big] \leq 2me^{-T/128} + m e^{-9R/32} \leq \frac{\delta}{2} + \frac{\delta}{2} = \delta.
\end{equation*}
The resulting algorithm uses a total number of copies
\begin{equation*}
    N \leq R \cdot 2kT \leq  O\bigg(\frac{\log^2(m/\delta)}{\eps^2}\bigg).
\end{equation*}
This completes the proof.
\end{proof}

It thus suffices to prove~\Cref{lem:randomized-integral-pgm}. The proof splits into three steps. Before we outline each step, we will need to introduce some useful notation related to Bernoulli random variables.

\subsection{Helpful notation}

Fix an observable $O_i$ throughout this subsection. 
We assume without loss of generality that $0 \preceq O_i \preceq I$, since estimating the expectation of $(I + O_i)/2 \succeq 0$ to $\pm \eps/2$ accuracy suffices to estimate $\tr(O_i \rho_x)$ to $\pm \eps$ accuracy. 

Our proof will compare the integral PGM's estimate of $\tr(O_i \cdot \rho_x)$ with the distribution obtained by measuring $O_i$ on independent copies of $\rho_x$. To make this comparison precise, we encode the expectation value $\tr(O_i \cdot \rho_x)$ as the bias of a Bernoulli random variable.
Equivalently, we can view this Bernoulli random variable as a diagonal quantum state, which we define below.

\begin{definition}[Auxiliary Bernoulli state]
    We  define the diagonal mixed state 
    \begin{equation*}
    B_x \coloneq \left(1 - \tr(O_i \rho_x)\right) \cdot\ketbra{0}{0} + \tr(O_i \rho_x)\cdot\ketbra{1}{1} = \Ber( \tr(O_i \cdot \rho_x) ),
    \end{equation*}
    and use $B_x^{\otimes k}$ for its $k$-fold tensor product:
    \begin{equation*}
        B_x^{\otimes k} \coloneq \Ber( \tr(O_i \cdot \rho_x) )^{\otimes k}.
    \end{equation*}
\end{definition}
Thus, $B_x^{\otimes k}$ corresponds to the classical probability distribution obtained by tossing $k$ coins with bias $\tr(O_i \cdot \rho_x)$. It will be useful to view it as a diagonal mixed state, since $B_x^{\otimes k}$ is the state that results when we measure each copy of the state $\rho_x^{\otimes k}$ using the POVM $\{O_i, I - O_i\}$. 

\begin{remark}
    Given $n+k$ copies of the state $\rho_x$, if we measure the last $k$ copies using the POVM $\{O_i, I-O_i\}$, the resulting state will be $\rho_x^{\otimes n} \otimes B_x^{\otimes k}$.
\end{remark}

\begin{definition}[Estimated auxiliary Bernoulli state]
\label{def:estimated_auxiliary_bernoulli}
    Let $\widehat{\bx}$ be a label estimate drawn from some distribution. We use $B_{\widehat{\bx}}^{(k)}$ to denote the \emph{estimated $k$-copy auxiliary Bernoulli state} $B_{\widehat{\bx}}^{\otimes k}$, i.e.\ the state
    \begin{equation}
        \label{eq:estimated-bernoulli-aux-state-generic}
        B_{\widehat{\bx}}^{(k)}
        \coloneq \sum_{\widehat{x}} \Pr\big[\widehat{\bx} = \widehat{x}\big] \cdot B^{\otimes k}_{\widehat{x}}.
    \end{equation}
    Throughout this paper, it will often be the case that $\widehat{\bx} = \calA(x, n)$, i.e.\ $\widehat{\bx}$ is the estimate for the label $x$ obtained after measuring $n$ copies of $\rho_x$ using~\Cref{alg:integral-pgm-estimator}. In that case, we will use the shorthand $B_{\widehat{\bx}}^{(k)}(p,n)$, which will make the dependence on the distribution over states $p$ and on the number of copies $n$ explicit:
    \begin{equation}
        \label{eq:estimated-bernoulli-aux-state-pgm}
        B_{\widehat{\bx}}^{(k)}(p, n)
        \coloneq \sum_{\widehat{x}} \Pr\big[\widehat{\bx} = \widehat{x} \, \big| \, n \big] \cdot B^{\otimes k}_{\widehat{x}}. 
    \end{equation}    
\end{definition}

Next, recall from the introduction that given a probability distribution $p$ over states $\rho_x$, we define the average $n$-copy state to be
\begin{equation*}
    S(p,n) \coloneq \sum_x p_x \cdot \rho_x^{\otimes n}.   
\end{equation*}
For the purposes of the proof below, we make the dependence on the distribution $p$ and the number of copies $n$ explicit. Moreover, we use $S^+$ to denote the average $n$-copy state with $k$ additional copies of the auxiliary Bernoulli state 
\begin{equation*}
    S^+(p, n, k) \coloneq \sum_x p_x \cdot \rho_x^{\otimes n} \otimes B_x^{\otimes k}. 
\end{equation*}
We will often omit the dependence on $k$ whenever it is clear from the context.

Finally, the following quantity will be useful:
\begin{equation*}
    \Delta_x(p, n) \coloneq D \big( \rho_x^{\otimes n} \otimes B_x^{\otimes k} \big\| S^+(p,n,k) \big) - D \big( \rho_x^{\otimes n} \big\| S(p,n) \big).
\end{equation*} 
While this quantity may seem mysterious at first, let us try to provide some intuition for it. 
Each term compares the state conditioned on $\bx=x$ with the average state when the label $\bx$ is drawn according to the prior $p$. 
The difference is that the first term includes both the $n$ copies of $\rho_x$ and the $k$ auxiliary Bernoulli registers, whereas the second term includes only the $n$ copies of $\rho_x$.
Thus, $\Delta_x(p,n)$ measures how much additional information about whether $\bx=x$ is provided by the Bernoulli registers, which themselves contain outcomes for measuring with the POVM $\{O_i, I-O_i\}$, the ideal measurement for the given observable $O_i$.
Equivalently, it measures how much relative entropy is lost when these registers are discarded.
Indeed, tracing out the last $k$ Bernoulli registers in the first term gives the states appearing in the second term. Data processing thus implies that $\Delta_x(p, n) \geq 0$.


What happens in the extreme case when $\Delta_x(p, n) = 0$? 
In this case, the additional Bernoulli registers provide no extra information for distinguishing $\bx = x$ from $\bx \sim p$, and thus can be reconstructed from $\rho_x^{\otimes n}$.
A bit more formally, equality in the data-processing inequality implies that (see, for example, the discussion around Equation (7.2.78) of~\cite{KW24}) there exists a channel $\calR$ that simultaneously maps
\begin{equation*}
    \calR(\rho_x^{\otimes n}) = \rho_x^{\otimes n} \otimes B_x^{\otimes k},
    \qquad
    \calR(S) = S^+.
\end{equation*}
Thus, starting only from $\rho_x^{\otimes n}$, the channel $\calR$ recovers the $k$ Bernoulli registers while preserving the original $n$-copy state. 
One can then measure the last $k$ copies to get an estimate for $\tr(O_i \rho_x)$, and would still have $n$ fresh copies of $\rho_x$. 
In other words, when $\Delta_x(p,n) = 0$, we can extract information about the observable $O_i$ without disturbing the state at all.

On the other hand, what happens when $\Delta_x(p,n)$ is large? 
In this case, the Bernoulli registers may contain information about whether $\bx=x$ that is not already present in the first $n$ copies of $\rho_x$.
Our proof will show that we can control the PGM's prediction error by this quantity. 
Small $\Delta_x(p,n)$ will imply that the PGM can accurately reproduce the Bernoulli registers, and our analysis will show that large values of $\Delta_x(p,n)$ cannot occur too often over the random choice of $n$. 


\subsection{Proof of \cref{lem:randomized-integral-pgm}}

The proof has three main ingredients.

\paragraph{Ingredient 1.} Upper bound the probability of error for each label $x$ using the relative entropy between the auxiliary Bernoulli states. In particular, we will prove the following statement.

\begin{lemma}
\label{lem:bernoulli-prediction}
Fix a label $x$, and let $\widehat{\bx}$ be a label estimator drawn according to some distribution. Recall our definition of the Bernoulli auxiliary state in~\Cref{def:estimated_auxiliary_bernoulli}.
Then for any observable $O_i$,
\begin{equation*}
    \Pr_{\wh{\bx}}\Big[
        \big|\tr(O_i\cdot \rho_{\wh{\bx}})- \tr(O_i\cdot \rho_x)\big|>\eps
    \Big]
    \leq 4e^{-k\epsilon^2/2} + \sqrt{
        \frac12 D\left( B_x^{\otimes k} \middle\| B_{\widehat{\bx}}^{(k)} \right)
    }.
\end{equation*}
\end{lemma}
This lemma holds for any distribution $\widehat{\bx}$ (not only when $\bx$ is the output of the integral PGM). Indeed, the statement and proof of this lemma are completely classical. The first term of the right hand side captures the probability that the average of $k$ Bernoulli random variables deviates more than $\epsilon/2$ from its expectation. The main tool used is Pinsker's inequality to translate the left-hand side to relative entropy.

Given the above lemma, we can write the probability of error of our algorithm for a fixed $x$ as
\begin{align*}
    \Pr_{\widehat{\bx} \sim \calA(x)} \Big[
        \big| \tr(O_i\cdot \rho_{\wh{\bx}})- \tr(O_i\cdot \rho_x) \big| >\eps 
    \Big] 
    & = \frac{1}{T} \sum_{t=T}^{2T-1} \Pr_{\wh{\bx} \sim \calA(x,kt)}\Big[
        \big|\tr(O_i\cdot \rho_{\wh{\bx}})- \tr(O_i\cdot \rho_x)\big| >\eps 
    \Big] \\
    & \leq  4e^{-k\epsilon^2/2} + \frac{1}{T} \sum_{t=T}^{2T-1} \sqrt{
        \frac12 D\left( B_x^{\otimes k} \middle\| B_{\widehat{\bx}}^{(k)}(p, kt) \right)}.
\end{align*}

\paragraph{Ingredient 2.} The next step is to upper bound the relative entropy appearing above using $\Delta_x$. In particular, we will show the following.\footnote{This lemma is a specialization of the data processing inequality that appears in Corollary 4.2 of~\cite{SBT16}. We include a proof of the special case we are interested in for completeness.}

\begin{lemma}
\label{lem:integral-pgm-recovery}
Fix a prior $p$ on $\calF$, a label $x\in\supp(p)$, and integers $n,k\geq1$.
Apply the integral PGM associated with $\{p_y,\rho_y^{\otimes n}\}_{y\in\calX}$ to $n$ copies of $\rho_x$, and let $\wh{\bx}$ denote its outcome.
Then
\begin{equation*}
    D\left(
        B_x^{\otimes k}
        \middle\| B_{\widehat{\bx}}^{(k)}(p, n)
    \right)
    \leq \Delta_x(p,n).
\end{equation*}
\end{lemma}
\noindent This makes the probability of error of our algorithm for a fixed $x$ at most
\begin{equation}
    \label{eq:equation-avg-case-1}
    \Pr_{\wh{\bx} \sim \calA(x)} \Big[
        \big| \tr(O_i\cdot \rho_{\wh{\bx}})- \tr(O_i\cdot \rho_x) \big| >\eps 
    \Big]
    \leq 4e^{-k\epsilon^2/2} + \frac{1}{\sqrt{2}\cdot T} \sum_{t=T}^{2T-1} \sqrt{\Delta_x(p, kt)}.
\end{equation}
It is worth noting that this is the only step in the proof that depends on the output of~\Cref{alg:integral-pgm-estimator}.

For an $i$ such that $\mathsf{Bad}_{O_i} \neq \varnothing$, we can define a distribution $p^{(i)}$ as $p$ conditioned on drawing an $x \in \mathsf{Bad}_{O_i}$. That is, take
\begin{equation*}
    p^{(i)}_x \coloneqq \begin{cases}
        p_x/\Pr[\mathsf{Bad}_{O_i}] & \text{if}~x \in \mathsf{Bad}_{O_i}, \\
        0     & \text{if}~x \in \mathsf{Good}_{O_i}.
    \end{cases}
\end{equation*}
Since $x \in \mathrm{Bad}_{O_i}$ if and only if $\Pr_{\wh{\bx} \sim \calA(x)} \big[\abs{\tr(O_i\cdot \rho_{\wh{\bx}})- \tr(O_i\cdot \rho_x)}>\eps \big] > 1/8$, we have the following lower bound:
\begin{equation*}
    \sum_{x} p^{(i)}_x \cdot \Pr_{\wh{\bx} \sim \calA(x)} \Big[
        \big| \tr(O_i\cdot \rho_{\wh{\bx}})- \tr(O_i\cdot \rho_x) \big| >\eps 
    \Big] > 1/8.
\end{equation*}
However, we also have an upper bound from \Cref{eq:equation-avg-case-1}:
\begin{align*}
    \sum_{x} p^{(i)}_x \cdot \Pr_{\wh{\bx} \sim \calA(x)} \Big[
        \big| \tr(O_i\cdot \rho_{\wh{\bx}})- \tr(O_i\cdot \rho_x) \big| >\eps 
    \Big] 
    &\leq 4e^{-k\epsilon^2/2} + \frac{1}{\sqrt{2}\cdot T} \sum_x p^{(i)}_x \sum_{t=T}^{2T-1} \sqrt{\Delta_x(p, kt)} \\
    & \leq 4e^{-k\epsilon^2/2} + \frac{1}{\sqrt{2}} \cdot \sqrt{\frac{1}{T} \sum_{t=T}^{2T-1}  \sum_x p^{(i)}_x \cdot \Delta_x(p, kt)}.
\end{align*}
The second inequality above follows from Jensen's inequality and the fact that the square root is concave.

\paragraph{Ingredient 3.} The final ingredient allows us to further upper bound the quantity we have obtained under the square root, providing us with a tractable upper bound.

\begin{lemma}
\label{lem:info-bound}
Let $p$ and $q$ be priors on $\calF$ with
$\supp(q)\subseteq\supp(p)$. Then, for every $k,T\geq1$,
\[
    \frac1T
    \sum_{t=T}^{2T-1}
    \sum_x q_x \cdot \Delta_x(p, kt) \leq \frac{D(q\,\|\,p)+\log 2}{T}.
\]
\end{lemma}
By substituting $p^{(i)}$ in the place of $q$, we can exactly calculate the relative entropy on the right-hand side as follows
\begin{equation*}
    D(p^{(i)}\,\|\,p) = \sum_x p^{(i)}_x \cdot \log\Big(\frac{p^{(i)}_x}{p_x}\Big) = -\log \big(\Pr[\mathrm{Bad}_{O_i}] \big).
\end{equation*}
Combining with the inequalities above, we have that
\begin{equation*}
    \frac{1}{8} \leq 4e^{-k\epsilon^2/2} + \frac{1}{\sqrt{2}} \cdot \sqrt{\frac{-\log\big(\Pr[\mathrm{Bad}_{O_i}]\big) + \log(2)}{T}}
\end{equation*}
Setting $k \geq 12/\epsilon^2$ will make $4e^{-k\epsilon^2/2} \leq  4e^{-6} \leq 1/16$. Thus
\begin{equation*}
     \frac{1}{8} \leq \frac{1}{16} + \frac{1}{\sqrt{2}} \sqrt{\frac{-\log(\Pr[\mathrm{Bad}_{O_i}]) + \log(2)}{T}}.
\end{equation*}
Rearranging gives
\begin{equation*}
    \Pr[\mathrm{Bad}_{O_i}] \leq 2e^{-T/128}.
\end{equation*}
This completes the proof. \qed










\subsection{Deferred proofs}
\label{sec:deferred-proofs}

In this section, we fill in the proofs of the lemmas stated in the proof of \Cref{lem:randomized-integral-pgm}. We begin with the ``first ingredient''.

\begin{lemma*}[\Cref{lem:bernoulli-prediction}, restated]
Fix a label $x$, and let $\widehat{\bx}$ be a label estimator drawn according to some distribution. Recall our definition of the Bernoulli auxiliary state in~\Cref{def:estimated_auxiliary_bernoulli}.
Then for any observable $O_i$,
\begin{equation*}
    \Pr_{\wh{\bx}}\Big[
        \big|\tr(O_i\cdot \rho_{\wh{\bx}})- \tr(O_i\cdot \rho_x)\big|>\eps
    \Big]
    \leq 4e^{-k\epsilon^2/2} + \sqrt{
        \frac12 D\left( B_x^{\otimes k} \middle\| B_{\widehat{\bx}}^{(k)} \right)
    }.
\end{equation*}
\end{lemma*}

\begin{proof}[Proof of \cref{lem:bernoulli-prediction}]

We will drop the $k$-superscripts for convenience. View the $\{B_{x'}\}_{x' \in \calL}$ and $B_{\widehat{\bx}}$ as classical probability distributions, and abbreviate $\tr(\rho_{x'} \cdot O_i) \to \theta_{x'}$ so that $B_{x'} = \Ber( \theta_{x'})^{\otimes k}$. Let $z$ be a binary string of length $k$, and write $\overline{z} \coloneq |z|/k$ for its average value. Then by Hoeffding's inequality,
\begin{equation} \label{eq:4.7_eq_1}
    \Pr_{\bz \sim B_{x'}} \Big[ \big|\overline{\bz} - \theta_{x'} \big| > \epsilon/2 \Big] \leq 2e^{-k\eps^2/2}. 
\end{equation} 
Consequently, if $| \theta_{x'} - \theta_x| > \epsilon$, then 
\begin{equation*}
    \Pr_{\bz \sim B_{x'}} \Big[ \big|\overline{\bz} - \theta_{x} \big| \leq \epsilon/2 \Big]  \leq \Pr_{\bz \sim B_{x'}} \Big[ \big|\overline{\bz} - \theta_{x'} \big| > \epsilon/2 \Big]  \leq 2e^{-k\eps^2/2}. 
\end{equation*}
On the other hand, if $|\theta_{x'} - \theta_x| \leq \eps$, we always have the trivial bound $\Pr_{\bz \sim B_{x'}} \big[ |\overline{\bz} - \theta_{x} | \leq \epsilon/2 \big] \leq 1$. Thus
\begin{align}
    \Pr_{\bz \sim B_{\widehat{\bx}}} \Big[ \big| \overline{\bz} - \theta_x\big| \leq \eps/2 \Big] & \leq \Pr_{\widehat{\bx}} \Big[ \big| \theta_{\widehat{\bx}} - \theta_{x} \big| > \eps \Big] \cdot \Big( 2e^{-k\eps^2/2} \Big) + \Pr_{\widehat{\bx}} \Big[ \big| \theta_{\widehat{\bx}} - \theta_{x} \big| \leq \eps \Big] \cdot 1 \nonumber \\
    & = 1 - \Big( 1 - 2e^{-k\eps^2/2} \Big) \cdot \Pr_{\widehat{\bx}} \Big[ \big| \theta_{\widehat{\bx}} - \theta_{x} \big| > \eps \Big] \nonumber \\
    & \leq 1 - \Pr_{\widehat{\bx}} \Big[ \big| \theta_{\widehat{\bx}} - \theta_{x} \big| > \eps \Big] + 2e^{-k\eps^2/2}. \label{eq:4.7_eq_2}
\end{align}
Hence, the total variation distance between the auxiliary Bernoulli states satisfies 
\begin{align*}
    \mathrm{TV}( B_x, B_{\widehat{\bx}}) & \geq \Pr_{\bz \sim B_{x}} \Big[ \big| \overline{\bz} - \theta_x\big| \leq \eps/2 \Big] - \Pr_{\bz \sim B_{\widehat{\bx}}} \Big[ \big| \overline{\bz} - \theta_x\big| \leq \eps/2 \Big] \\
    & \geq \Big( 1 - 2e^{-k\eps^2/2} \Big) - \Big( 1 - \Pr_{\widehat{\bx}} \Big[ \big| \theta_{\widehat{\bx}} - \theta_{x} \big| > \eps \Big] + 2e^{-k\eps^2/2}\Big) \tag{\Cref{eq:4.7_eq_1,eq:4.7_eq_2}} \\
    & = \Pr_{\widehat{\bx}} \Big[ \big| \theta_{\widehat{\bx}} - \theta_{x} \big| > \eps \Big] - 4e^{-k\epsilon^2/2}.
\end{align*}
Rearranging, and using Pinsker's inequality, we obtain:
\begin{equation*}
    \Pr_{\widehat{\bx}} \Big[ \big| \theta_{\widehat{\bx}} - \theta_{x} \big| > \eps \Big] \leq 4e^{-k\epsilon^2/2} + \mathrm{TV}( B_x, B_{\widehat{\bx}}) \leq 4e^{-k\epsilon^2/2} + \sqrt{ \frac{1}{2} D(B_x \| B_{\widehat{\bx}})}. 
\end{equation*}
This completes the proof. \qedhere

\end{proof}

We continue with the ``second ingredient''. 

\begin{lemma*}[\Cref{lem:integral-pgm-recovery}, restated]
Fix a prior $q$ on $\calF$, a label $x\in\supp(q)$, and integers $n,k\geq1$.
Apply the integral PGM associated with $\{q_y,\rho_y^{\otimes n}\}_{y\in\calX}$ to $n$ copies of $\rho_x$, and let $\wh{\bx}$ denote its outcome.
Then
\begin{equation*}
    D\left(
        B_x^{\otimes k}
        \middle\| B_{\widehat{\bx}}^{(k)}(q, n)
    \right)
    \leq \Delta_x(p,n).
\end{equation*}
\end{lemma*}

\begin{proof}[Proof of \cref{lem:integral-pgm-recovery}]
Write $S\coloneqq S(q,n) = \sum_y q_y\rho_y^{\otimes n}$.
Since each $B_y^{\otimes k}$ is diagonal in the computational basis, for each
$z\in\{0,1\}^k$ define
\[
    X_z
    \coloneqq
    \sum_y
    q_y \cdot 
    \langle z|B_y^{\otimes k}|z\rangle \cdot 
    \rho_y^{\otimes n}.
\]
Then we have
\[
    S^+(q,n) = \sum_z X_z\otimes\ketbra{z}{z}.
\]

We first relate $X_z$ to the distribution predicted by the PGM. By the definition of $B_{\wh{\bx}}^{(k)}(q,n)$, we have
\begin{equation} \label{eq:second_ingredient_proof_eq_1}
\bra{z} B_{\widehat{\bx}}^{(k)}(q,n)\ket{z}
    = \sum_y
    \Pr_{\wh{\bx} \sim \calA(x,n)}[\wh{\bx}=y]
    \cdot \bra{z}B_y^{\otimes k}\ket{z}= \tr\!\big( \rho_x^{\otimes n} \cdot \calD_S(X_z) \big),
\end{equation}
where we have used the definition of the integral PGM and the linearity of
$\calD_S$.

Now fix $z$ such that $\bra{z}B_x^{\otimes k}\ket{z}>0$.
Since $x\in\supp(q)$, we have $q_x>0$, and hence $X_z \succeq q_x \cdot \bra{z}B_x^{\otimes k}\ket {z} \cdot \rho_x^{\otimes n}$.
Thus $\supp(\rho_x^{\otimes n}) \subseteq \supp(X_z)$.
Moreover,
\[
    X_z \preceq \sum_y q_y \rho_y^{\otimes n} = S,
\]
so $\supp(X_z)\subseteq\supp(S)$.
We may therefore apply \cref{lem:integral_PGM_inequality} with
$\sigma=\rho_x^{\otimes n}$, obtaining
\[
    \tr\!\big( \rho_x^{\otimes n} \cdot (\log X_z-\log S) \big) \leq \log \tr \big( \rho^{\otimes n}_x \cdot \calD_S(X_z) \big) = \log \bra{z} B_{\widehat{\bx}}^{(k)}(q,n) \ket z. 
\]
Here, we have also used \Cref{eq:second_ingredient_proof_eq_1}. 
Multiplying by $\langle z|B_x^{\otimes k}|z\rangle$, using the fact that $B_x^{\otimes k}$ and $B_{\wh{\bx}}^{(k)}$ are both states diagonal in the $z$-basis, and summing over $z$ then gives
\begin{align*}
    \tr\!\Big(
        B_x^{\otimes k} \cdot \log B_{\widehat{\bx}}^{(k)}(q,n)
    \Big)
    &\geq
    \Big(\sum_z \bra{z}B_x^{\otimes k}\ket z \cdot 
    \tr\big(\rho_x^{\otimes n} \cdot \log X_z\big) \Big) - \tr\big(\rho_x^{\otimes n} \cdot \log S\big).
\end{align*}
Since $S^+(q,n)=\sum_z X_z\otimes\ketbra{z}{z}$ is block-diagonal in the
$z$-register,
\[
    \sum_z \bra z B_x^{\otimes k} \ket z  \cdot \tr\big(\rho_x^{\otimes n} \cdot \log X_z\big)
    = \tr\!\Big( \big(\rho_x^{\otimes n}\otimes B_x^{\otimes k}\big) \cdot \log S^+(q,n)
    \Big).
\]
Therefore,
\begin{align*}
    D\left(
        B_x^{\otimes k}
        \middle\|
        B_{\widehat{\bx}}^{(k)}(q,n)
    \right) & = \tr \big( B^{\otimes k}_{x} \cdot \log B^{\otimes k}_x \big) - \tr \Big( B^{\otimes k}_x \cdot \log B^{(k)}_{\wh{\bx}}(q,n) \Big)\\
    & \leq \tr \big( B^{\otimes k}_{x} \cdot \log B^{\otimes k}_x \big) -\Big(\sum_z \bra{z}B_x^{\otimes k}\ket z \cdot 
    \tr\big(\rho_x^{\otimes n} \cdot \log X_z\big) \Big) + \tr\big(\rho_x^{\otimes n} \cdot \log S\big) \\
    & = \tr \big( B^{\otimes k}_{x} \cdot \log B^{\otimes k}_x \big) -\tr\!\big( \big(\rho_x^{\otimes n}\otimes B_x^{\otimes k}\big) \cdot \log S^+(q,n)
    \big) + \tr\big(\rho_x^{\otimes n} \cdot \log S\big).
\end{align*}
Finally, using
\[
    \log(\rho_x^{\otimes n}\otimes B_x^{\otimes k})
    = \log(\rho_x^{\otimes n})\otimes I + I\otimes\log(B_x^{\otimes k}),
\]
the right-hand side is exactly
\begin{align*}
    \Delta_q(x,n) & = D\left(
        \rho_x^{\otimes n}\otimes B_x^{\otimes k} \middle\| S^+(q,n)
    \right)
    - D\left(
        \rho_x^{\otimes n} \middle\| S(q,n)
    \right) \\
    & = \tr \Big( \rho_x^{\otimes n}\otimes B_x^{\otimes k} \cdot \log \big(\rho_x^{\otimes n}\otimes B_x^{\otimes k}\big) \Big) - \tr \Big( \rho_x^{\otimes n}\otimes B_x^{\otimes k} \cdot \log S^+ \Big) - \tr \Big( \rho_x^{\otimes n} \cdot \log \rho_x^{\otimes n} \Big) + \tr\Big( \rho_x^{\otimes n} \cdot \log S \Big),
\end{align*}
since 
\begin{equation*}
     \tr \Big( \rho_x^{\otimes n}\otimes B_x^{\otimes k} \cdot \log \big(\rho_x^{\otimes n}\otimes B_x^{\otimes k}\big) \Big) = \tr \Big( \rho_x^{\otimes n} \cdot \log \rho_x^{\otimes n} \Big) + \tr \Big( B^{\otimes k}_x \cdot \log B^{\otimes k}_x \Big).
\end{equation*}
This completes the proof.
\end{proof}

Finally, we fill in the proof of the ``third ingredient". This requires a little more effort than the previous proofs. We start by restating the lemma. 

\begin{lemma*}[\Cref{lem:info-bound}]
Let $p$ and $q$ be priors on $\calF$ with
$\supp(q)\subseteq\supp(p)$. Then, for every $k,T\geq1$,
    \begin{equation*}
        \frac1T
        \sum_{t=T}^{2T-1}
        \sum_x q_x \cdot \Delta_x(p, kt) \leq \frac{D(q\,\|\,p)+\log 2}{T}.
    \end{equation*}
\end{lemma*}

We will prove a bound when the number of samples is fixed at some value $s$, which will imply the desired result. 

\begin{lemma}
    \label{prop:info-bound-per-t}
    Let $p$ and $q$ be priors on $\calF$ with $\supp(q) \subseteq\ supp(p)$. Then, for every $s, k\geq1$,
    \begin{equation*}
        \sum_x q_x \cdot \Delta_x(p, s) \leq D\big(S(q, s+k) \,\big\|\, S(p, s+k)\big) - D\big(S(q, s) \big\| S(p, s)\big) + \log(\frac{s + k + 1}{s + 1}).
    \end{equation*}
\end{lemma}

The proof of~\Cref{lem:info-bound} follows via telescoping.
\begin{proof}[Proof of \Cref{lem:info-bound}, assuming \Cref{prop:info-bound-per-t}]
    We apply~\Cref{prop:info-bound-per-t} to each term of the summation over $t$ to get
    \begin{align*}
        &\frac1T \sum_{t=T}^{2T-1} \sum_x q_x \cdot \Delta_x(p, kt) \\
        &\leq \frac1T \bigg[\sum_{t=T}^{2T-1} \Big( D\big(S(q, kt+k) \,\big\|\, S(p, kt+k)\big) - D\big(S(q, kt) \,\big\|\, S(p, kt)\big) \Big) + \sum_{t=T}^{2T-1} \log(\frac{kt + k + 1}{kt + 1}) \bigg] \\
        &= \frac1T \bigg[D\big(S(q, 2kT) \,\big\|\, S(p, 2kT)\big) - D\big(S(q, kT) \,\big\|\, S(p, kT)\big) + \log(\frac{2kT + 1}{kT + 1}) \bigg] \\
        &\leq \frac1T \Big[D\big(S(q, 2kT)\,\big\|\, S(p, 2kT)\big) + \log2 \Big].
    \end{align*}
    Now consider the classical-quantum channel $\calR$ which maps the classical label $x$ to $\rho_x^{\otimes 2kT}$. In particular, the channel will map the distribution $q$ of labels to the quantum state $S(q, 2kT)$ and similarly for $p$. Data processing then implies that
    \begin{equation*}
        D\big(S(q, 2kT) \,\big\|\, S(p, 2kT)\big)
        = D\big(\calR(q) \,\big\|\, \calR(p)\big)
        \leq D(q \,\|\, p ).
    \end{equation*}
    This completes the proof.
\end{proof}
Before we turn our attention to proving~\Cref{prop:info-bound-per-t}, we will need to collect a couple entropic identities. 

\begin{identity}[A first entropic identity]
    \label{prop:donalds-identity}
    Let $\{\sigma_x\}_{x \in \calL}$ be a family of quantum states and $p$ be a distribution over the labels $\calL$.
    For any quantum state $\tau$, it holds that
    \begin{equation*}
        \sum_x p_x \cdot D(\sigma_x\,\|\,\tau) = \sum_x p_x \cdot D\Big(\sigma_x\,\Big\|\,\sum_y p_y \cdot \sigma_y\Big) + D\Big(\sum_y p_y \cdot \sigma_y\,\Big\|\, \tau\Big).
    \end{equation*}
\end{identity}

\begin{proof} 
    We expand the relative entropies. The left-hand side is
    \begin{equation} \label{eq:donald_1}
        \sum_x p_x \cdot D(\sigma_x\,\|\,\tau) = \Tr\bigg[\sum_x p_x \cdot \sigma_x \cdot \log (\sigma_x) - \sum_x p_x \cdot \sigma_x \cdot \log (\tau) \bigg]. 
    \end{equation}
    Letting $S = \sum_y p_y \cdot \sigma_y$, the first term on the right-hand side is
    \begin{equation} \label{eq:donald_2}
        \sum_x p_x \cdot D\Big(\sigma_x\,\Big\|\,\sum_y p_y \cdot \sigma_y\Big) = \Tr\bigg[\sum_{x} p_x \cdot \sigma_x \cdot \log (\sigma_x) - \sum_{x} p_x \cdot \sigma_x \cdot \log (S) \bigg],
    \end{equation}
    and the second term on the right-hand side is
    \begin{equation} \label{eq:donald_3}
        D\Big(\sum_y p_y \cdot \sigma_y\,\Big\|\, \tau\Big) = \Tr\bigg[\sum_x p_x \cdot \sigma_x \cdot \log(S) - \sum_x p_x \cdot \sigma_x \cdot \log(\tau)\bigg]. 
    \end{equation}
    We see that adding the right-hand sides of \Cref{eq:donald_2,eq:donald_3} yields the right-hand side \Cref{eq:donald_1}. This completes the proof. 
\end{proof}

\begin{identity}[A second entropic identity]
    \label{prop:rel-entropy-wrt-entropy}
    Let $\{\sigma_x\}_{x \in \calL}$ be a family of quantum states and $p$ be a distribution over the labels $\calL$. It holds that
    \begin{equation*}
        \sum_x p_x \cdot D\Big(\sigma_x\,\Big\|\,\sum_y p_y\cdot \sigma_y\Big) = H\Big(\sum_x p_x\cdot \sigma_x\Big) - \sum_x p_x \cdot H(\sigma_x).
    \end{equation*}
\end{identity}

\begin{proof}
We expand the entropies. Letting $S = \sum_y p_y \cdot \sigma_y$, the left-hand side is
    \begin{equation} \label{eq:second_identity_1}
        \sum_x p_x \cdot D\Big(\sigma_x\,\Big\|\,\sum_y p_y\cdot \sigma_y\Big) = \Tr\bigg[\sum_x p_x \cdot \sigma_x \cdot \log(\sigma_x) - \sum_x p_x \cdot \sigma_x \cdot \log(S)\bigg]. 
    \end{equation}
    Meanwhile, the first term on the right-hand side is
    \begin{equation} \label{eq:second_identity_2}
        H\Big(\sum_x p_x\cdot \sigma_x\Big) = -\Tr\bigg[\sum_x p_x \cdot \sigma_x \cdot \log(S)\bigg],
    \end{equation}
    and the second term on the right-hand side is
    \begin{equation} \label{eq:second_identity_3}
        \sum_x p_x \cdot H(\sigma_x) = -\Tr\bigg[\sum_x p_x \cdot \sigma_x \cdot \log(\sigma_x)\bigg]. 
    \end{equation}
    We see that subtracting the right-hand side of \Cref{eq:second_identity_3} from the right-hand side of \Cref{eq:second_identity_2} yields the right-hand side \Cref{eq:second_identity_1}. This completes the proof. 
\end{proof}

\begin{proof}[Proof of \Cref{prop:info-bound-per-t}]
    We first apply~\Cref{prop:donalds-identity} to each term separately:
\begin{align*}
    \sum_x q_x \cdot \Delta_x(p, s)
    &= \sum_x q_x \cdot D\left( \rho_x^{\otimes s}\otimes B_x^{\otimes k} \middle\| S^+(p,s) \right) - \sum_x q_x \cdot D\left( \rho_x^{\otimes s} \middle\| S(p,s) \right) \nonumber \\
    &= \bigg(\sum_x q_x \cdot D\left( \rho_x^{\otimes s}\otimes B_x^{\otimes k} \middle\| S^+(q,s) \right) + D(S^+(q, s) \| S^+(p, s))\bigg) \nonumber \\
    &\qquad - \bigg(\sum_x q_x \cdot D\left( \rho_x^{\otimes s} \middle\| S(q,s) \right) + D(S(q, s) \| S(p, s))\bigg) \nonumber \\
    &= \sum_x q_x \cdot D\left( \rho_x^{\otimes s}\otimes B_x^{\otimes k} \middle\| S^+(q,s) \right) - \sum_x q_x \cdot D\left( \rho_x^{\otimes s} \middle\| S(q,s) \right) \nonumber \\
    &\qquad + D(S^+(q, s) \| S^+(p, s)) - D(S(q, s) \| S(p, s)). \label{eq:third_ingredient_eq6}
\end{align*}
We will keep the fourth term as is. The third term can be bounded above using data processing, as follow. Let $\calR$ be the channel that takes $s + k$ copies of a state and measures the last $k$ copies using the POVM $\{O_i, I - O_i\}$. Then
\begin{equation*}
    D\big(S^+(q, s) \big\| S^+(p, s)\big) = D\big(\calR(S(q, s+k))\big\|\calR(S(p, s+k))\big)
    \leq D\big(S(q, s+k) \big\| S(p, s+k)\big). \label{eq:third_ingredient_eq5}
\end{equation*}

It remains to bound the difference of the first two terms. In that case, we use~\Cref{prop:rel-entropy-wrt-entropy} and the additivity of entropy over tensor products to get
\begin{align} 
    &\sum_x q_x \cdot D\left( \rho_x^{\otimes s}\otimes B_x^{\otimes k} \middle\| S^+(q,s) \right) - \sum_x q_x \cdot D\left( \rho_x^{\otimes s} \middle\| S(q,s) \right) \nonumber \\
    &\qquad = \Big( H(S^+(q,s)) - \sum_x q_x \cdot H(\rho_x^{\otimes s}\otimes B_x^{\otimes k}) \Big) 
    - \Big( H(S(q,s)) - \sum_x q_x \cdot H(\rho_x^{\otimes s}) \Big) \nonumber \\
    & \qquad = \Big( H(S^+(q,s)) - \sum_x q_x \cdot \big( s \cdot H(\rho_x) + k \cdot H(B_x) \big) \Big) - \Big( H(S(q,s)) - \sum_x q_x \cdot s \cdot H(\rho_x)\Big) \nonumber \\
     &\qquad = H(S^+(q,s)) - H(S(q,s)) - \sum_x q_x \cdot k \cdot H(B_x) \nonumber \\
     & \qquad = H(S^+(q,s)) - H(S(q,s)) - k \cdot \sum_x q_x\cdot h\big( \tr(O_i \cdot \rho_x) \big). \label{eq:third_ingredient_eq_1}
\end{align}
Now, we claim the bound
\begin{equation} \label{eq:third_ingredient_eq_2}
    H(S^+(q,s)) - H(S(q,s)) \leq H \Big(\sum_x q_x \cdot B^{\otimes (n+k)}_{x} \Big) -  H \Big(\sum_x q_x \cdot B^{\otimes n}_{x} \Big).
\end{equation}
To see this, consider the relative entropy
\begin{align*}
        &D \Big( S^+(q,n) \Big\| S(q,n) \otimes \frac{I^{\otimes k}}{2^k} \Big) \\
        ={}& \Tr \bigg[ \Big( \sum_x q_x \cdot \rho_x^{\otimes n} \otimes B_x^{\otimes k} \Big) \cdot \bigg( \log  \Big( \sum_x q_x \cdot \rho_x^{\otimes n} \otimes B_x^{\otimes k} \Big) - \log  \Big( \sum_x q_x \cdot \rho_x^{\otimes n} \otimes \frac{I^{\otimes k}}{2^k} \Big)  \bigg)\bigg] \\
        ={}& -H \Big( S^+(q,n) \Big) + H( S(q,n) \Big) + k \log(2). 
    \end{align*}
    If we take these states and measure the first $n$ factors with the measurement $\{O_i, I - O_i\}$, then we get instead
    \begin{equation*}
        D \Big( \sum_x q_x \cdot B^{\otimes (n+k)}_{x} \Big\| \sum_x q_x \cdot B_x^{\otimes n} \otimes \frac{I^{\otimes k}}{2^k} \Big) = -H \Big(\sum_x q_x \cdot B^{\otimes (n+k)}_{x} \Big) +  H \Big(\sum_x q_x \cdot B^{\otimes n}_{x} \Big) + k \log (2). 
    \end{equation*}
    Thus, by data processing, after canceling the $k\log(2)$ terms and rearranging, we obtain \Cref{eq:third_ingredient_eq_2}. Substitution of that bound into \Cref{eq:third_ingredient_eq_1} then yields
    \begin{align*}
        \eqref{eq:third_ingredient_eq_1} & \leq  \bigg(H \Big(\sum_x q_x \cdot B^{\otimes (n+k)}_{x} \Big) -  H \Big(\sum_x q_x \cdot B^{\otimes n}_{x} \Big)\bigg) - k \cdot \sum_{x} q_x \cdot h \big( \tr(O_i \cdot \rho_x) \big) \\
        & \leq \bigg( k \cdot \sum_{x} q_x \cdot h \big( \tr(O_i \cdot \rho_x) \big) + \log\Big( \frac{n+k+1}{n+1} \Big) \bigg) - k \cdot \sum_{x} q_x \cdot h \big( \tr(O_i \cdot \rho_x) \big) \\
        & = \log\Big( \frac{n+k+1}{n+1} \Big).
    \end{align*}
    The second inequality follows from \Cref{lem:entropic_(di)lemma} below, with $\btheta \leftarrow \tr(O_i \cdot \rho_{\bx})$, and with $\bx \sim q$. This completes our bound on the difference of the first two terms. Putting everything together finally yields
    \begin{equation*}
         \sum_x q_x \cdot \Delta_x(p, s) \leq  D\big(S(q, s+k) \big\| S(p, s+k)\big) - D\big(S(q, s) \big\| S(p, s)\big) + \log\Big( \frac{n+k+1}{n+1} \Big),
    \end{equation*}
    which is what we wanted to show. \qedhere
    \end{proof}

All that remains is an inequality concerning Shannon entropies, which we appealed to in the previous lemma. 




\begin{lemma}[A useful inequality for Shannon entropies] \label{lem:entropic_(di)lemma}
    Let $\btheta$ be any random variable on $[0,1]$, and let $m$ and $n$ be integers such that $1 \leq m \leq n$. For each $i \in [n]$, let $\bz_i \sim \Ber( \btheta)$. Then for any $m \leq n$, 
    \begin{equation*}
        H(\bz_{1}, \dots, \bz_n) - H(\bz_{1}, \dots, \bz_m) \leq (n-m) \cdot \E_{\btheta} \big[ h(\btheta) \big] + \log \Big( \frac{n+1}{m+1} \Big). 
    \end{equation*}
\end{lemma}

\begin{proof}
    We first reduce to the $n = m+1$ case, since if that holds, then the general case does too, by telescoping:
    \begin{align*}
        H(\bz_{1}, \dots, \bz_n) - H(\bz_{1}, \dots, \bz_m) & = \sum_{k = m}^{n-1} \Big( H(\bz_1, \dots, \bz_{k+1}) - H(\bz_1, \dots, \bz_k) \Big) \\
        & \leq \sum_{k=m}^{n-1} \Big( \E_{\btheta} \big[ h(\btheta) \big] + \log \Big( \frac{k+2}{k+1}\Big) \Big) \\
        & = (n-m) \cdot \E_{\btheta} \big[ h(\btheta) \big] + \log \Big( \frac{n+1}{m+1} \Big).
    \end{align*}

    So, take $n = m+1$. The one-step inequality follows from Equation (17) in~\cite[Proposition 4]{MG22}. In particular, specializing their result to a binary alphabet, they show that after observing $m$ samples $\bz_1, \dots, \bz_m \sim \Ber(\theta)$ from a Bernoulli distribution, the \emph{Laplace estimator}
    \begin{equation*}
        \widehat{\btheta}_m = \frac{\bz_1 + \dots \bz_m +1}{m + 2}
    \end{equation*}
    satisfies
    \begin{equation*}
        \E\big[D(\Ber(\theta) \,\|\, \Ber(\widehat{\btheta}_m))\big] \leq \log(\frac{m+2}{m+1}).
    \end{equation*}
    Averaging over $\btheta$ then yields
    \begin{equation}
        \label{eq:bern-proof-gyros}
        \E\big[D(\Ber(\btheta) \,\|\, \Ber(\widehat{\btheta}_m))\big] \leq \log(\frac{m+2}{m+1}).
    \end{equation}
    Now we will expand the relative entropy as
    \begin{equation}
        \label{eq:bern-proof-pita}
        \E\big[D(\Ber(\btheta) \,\|\, \Ber(\widehat{\btheta}_m))\big] = -\E \big[h(\btheta)\big] + \E\Big[-\btheta \cdot \log \big(\widehat{\btheta}_m\big) - (1 - \btheta) \cdot \log\big(1 - \widehat{\btheta}_m\big)\Big].
    \end{equation}
    Let us now focus on the second term and use the law of total expectation. By conditioning on $\bz_1, \dots, \bz_m$, we rewrite this term as
    \begin{equation*}
        \E_{\bz_1, \dots, \bz_m} \Big[\E \big[-\btheta \cdot \log \big(\widehat{\btheta}_m\big)  \, \big| \, \bz_1, \dots, \bz_m\big] - \E\big[(1 - \btheta) \cdot \log\big(1 - \widehat{\btheta}_m\big) \,\big|\, \bz_1, \dots, \bz_m\big]\Big].
    \end{equation*}
    Since $\widehat{\btheta}_m$ is constant when conditioned on the samples $\bz_1, \dots, \bz_m$, this is then equal to
    \begin{equation*}
        \E_{\bz_1, \dots, \bz_m} \Big[-\E \big[\btheta \big| \bz_1, \dots, \bz_m\Big] \cdot \log \big(\widehat{\btheta}_m\big) - \Big(1 - \E\big[\btheta\mid \bz_1, \dots, \bz_m\big]\Big)\cdot \log\big(1 - \widehat{\btheta}_m\big)\Big].
    \end{equation*}
    If we now set $\bphi = \E \big[\btheta \big| \bz_1, \dots, \bz_m\big]$, then this can be rewritten as
    \begin{align}
        &\E_{\bz_1, \dots, \bz_m} \Big[-\bphi \cdot \log \big(\widehat{\btheta}_m\big) - (1-\bphi) \cdot \log\big(1 - \widehat{\btheta}_m\big)\Big] \nonumber \\
        &\qquad =
        \E_{\bz_1, \dots, \bz_m}[h(\bphi)] + \E_{\bz_1, \dots, \bz_m} \Big[D(\Ber(\bphi) \,\|\, \Ber(\widehat{\btheta}_m) \Big] \geq \E_{\bz_1, \dots, \bz_m}[h(\bphi)]. \label{eq:bern-proof-souvlaki}
    \end{align}
    The inequality uses the nonnegativity of the relative entropy.
    Finally, observe that 
    \begin{equation*}
        \bz_{m+1} \mid \bz_1, \dots, \bz_m \sim \Ber\Big(\E_{\btheta} \big[\btheta \big| \bz_1, \dots, \bz_m\big]\Big) = \Ber(\bphi).
    \end{equation*}
    Taking the entropy on both sides and averaging over $\bz_1, \dots, \bz_m$ gives:
    \begin{equation*}
        \E_{\bz_1, \dots, \bz_m}[h(\bphi)]
        = H(\bz_{m+1} \mid \bz_1, \dots, \bz_m) = H(\bz_1, \dots, \bz_{m+1}) - H(\bz_1, \dots, \bz_{m}).
    \end{equation*}
    Combining the above with~\Cref{eq:bern-proof-gyros,eq:bern-proof-pita,eq:bern-proof-souvlaki} yields the desired bound. \qedhere 
    

\end{proof}

\subsection{Worst-case shadow tomography}

The main conceptual message of this paper is that it can be easier to solve worst-case shadow tomography by first passing through the average case. 
Indeed, once we have an algorithm that succeeds for every prior, the corresponding worst-case result follows almost immediately from the minimax theorem. 
We can therefore now prove the main theorem, restated below, as a direct consequence of \cref{thm:average-case-shadow-tomography}.

\WorstCaseShadow*

\begin{proof}
    Apply \cref{thm:average-case-shadow-tomography} for every prior and then apply the minimax theorem, exactly as described in \cref{subsec:background-minimax}.
\end{proof}

\section*{Acknowledgments}
We would like to thank Omar Alrabiah and Harshit Sikchi for helpful conversations.
S.G.\ is supported by a Herman Goldstine Memorial Postdoctoral Fellowship. 
A.P.\ is supported by DARPA under Agreement No.\ HR00112020023.
J.S.\ and J.W.\ are supported by the NSF CAREER award CCF-233971.
E.T.\ is supported by the Miller Institute for Basic Research in Science, University of California, Berkeley. 

\bibliographystyle{alpha}
\bibliography{wright}

\appendix

\section{A sharper boosting analysis}
\label{app:boost}

We present a proof of \cref{thm:average-case-shadow-tomography} achieving the claimed $O(\log(m)\log(m/\delta)/\eps^2)$ copy complexity bound. 

\begin{proof}[Proof of \cref{thm:average-case-shadow-tomography}]
Fix parameter $T=\Theta(\log m)$ large enough so that $2m e^{-T/128}\leq 1/8$.
For any prior $q$ on $\calF$, run
\cref{alg:integral-pgm-estimator} with these parameters, and let
\[
    e_q(x) \coloneqq \max_{i\in[m]} \Pr\left[ \abs{\wh{\bo}_i-\tr(O_i\rho_x)}>\eps \mid \bx=x
    \right]
\]
denote the largest conditional error probability among the $m$ observable estimates on input $\rho_x$.
Here, the probability is over the randomness of \cref{alg:integral-pgm-estimator} when run with prior $q$.
By \cref{lem:randomized-integral-pgm} and a union bound over the $m$
observables,
\[
    \Pr_{\bx\sim q}\left[e_q(\bx)>\frac18\right]
    \leq 2m e^{-T/128}
    \leq\frac18.
\]
Since always $e_q(x)\leq1$, it follows that
\begin{equation}
\label{eq:weak-pgm-guarantee}
    \E_{\bx\sim q}[e_q(\bx)]
    \leq\frac14.
\end{equation}

Our goal now is to amplify this guarantee. We do so using a standard boosting argument, iteratively reweighting the prior toward states on which the previous estimators have larger error.
For an odd integer $R$, we will construct a sequence of priors $q_1,\ldots,q_R$, where each new prior places more weight on states for which the previous estimators have larger error. 
For each $r\in[R]$, we run \cref{alg:integral-pgm-estimator} with prior $q_r$, obtaining a vector of estimates $\wh{\bo}^{(r)}$. 
Our final output is the coordinatewise median of $\wh{\bo}^{(1)},\ldots,\wh{\bo}^{(R)}$.

Initialize the weights by setting $w_0(x)=p_x$ for each $x\in\calX$. 
For each $r=1,\ldots,R$, define
\[
    Z_{r-1}\coloneqq\sum_x w_{r-1}(x),
    \qquad
    q_r(x)\coloneqq\frac{w_{r-1}(x)}{Z_{r-1}}.
\]
Thus, $q_r$ is the distribution obtained by normalizing the current weights. Recall that $e_{q_r}(x)$ is the largest conditional error probability of the estimator associated with the prior $q_r$ on the state $\rho_x$. We use these error probabilities to update the weights according to
\[
    w_r(x) \coloneqq w_{r-1}(x)\bigl(1+2e_{q_r}(x)\bigr).
\]
In particular, states on which the current estimator has larger error receive more weight in the next round.

By \cref{eq:weak-pgm-guarantee},
\[
    Z_r = Z_{r-1}\left(1+2\E_{\bx\sim q_r}[e_r(\bx)]\right)
    \leq \frac32 Z_{r-1},
\]
and hence
\begin{equation}
\label{eq:weight-growth}
    Z_R\leq\left(\frac32\right)^R.
\end{equation}

We now show that the coordinatewise median of the $R$ estimators succeeds with high probability. Fix a state $\rho_x$ and an observable $O_i$. For each $r\in[R]$, define $\bF_{\!r}$ to be the indicator of the failure event that the $r$th estimator is inaccurate on $O_i$:
\[
\bF_{\!r} \coloneqq \mathbf 1\left\{ \abs{\wh{\bo}^{(r)}_i-\tr(O_i\rho_x)}>\eps \right\}.
\]
By definition, we have $\Pr[\bF_{\!r}=1\mid \bx=x] \leq e_{q_r}(x)$.
Since the $R$ estimators are applied independently to disjoint sets of copies, the random variables $\bF_{\!1},\ldots,\bF_{\!r}$ are independent conditioned on $\bx=x$.

If the median estimate for $O_i$ is inaccurate, then at least half of the individual estimates must be inaccurate. Hence
\begin{align*}
    \Pr\left[
        \abs{\wh{\bo}_i-\tr(O_i\rho_x)}>\eps \mid\bx=x
    \right]
    &\leq
    \Pr\left[
        \sum_{r=1}^R \bF_{\!r}\geq \frac{R}{2} \mid\bx=x
    \right] \\
    &\leq
    3^{-R/2}
    \E\left[
        3^{\sum_{r=1}^R \bF_{\!r}} \mid \bx=x
    \right] \\
    &=
    3^{-R/2}
    \prod_{r=1}^R
    \E\left[3^{\bF_{\!r}}\mid\bx=x\right] \\
    &\leq
    3^{-R/2}
    \prod_{r=1}^R
    \bigl(1+2e_{q_r}(x)\bigr).
\end{align*}
Here, the second inequality is Markov's inequality, the third line uses the independence of $\bF_{\!r}$'s, and in the last line we used
\[
    \E[3^{\bF_{\!r}}\mid\bx=x] = 1+2\Pr[\bF_{\!r}=1\mid\bx=x] \leq 1+2e_{q_r}(x).
\]

By the definition of the weights, we have
\[
    w_R(x)
    = p_x\prod_{r=1}^R \bigl(1+2e_{q_r}(x)\bigr).
\]
Thus, averaging over $\bx\sim p$, we obtain
\begin{align*}
    \Pr\left[
        \abs{\wh{\bo}_i-\tr(O_i\brho)}>\eps
    \right]
    &\leq
    3^{-R/2}
    \sum_x
    p_x\prod_{r=1}^R
    \bigl(1+2e_{q_r}(x)\bigr) \\
    &=
    3^{-R/2} Z_R \\
    &\leq
    \left(\frac{\sqrt{3}}{2}\right)^R,
\end{align*}
where the last inequality follows from \cref{eq:weight-growth}.

Finally, taking a union bound over the $m$ observables gives
\[
    \Pr\left[
        \max_{i\in[m]}
        \abs{\wh{\bo}_i-\tr(O_i\brho)}>\eps
    \right]
    \leq
    m\left(\frac{\sqrt{3}}{2}\right)^R.
\]
Taking $R=O(\log(m/\delta))$ therefore makes this probability at most $\delta$.

Each invocation of \cref{alg:integral-pgm-estimator} uses at most
\[
    k(2T-1) = O\left(\frac{\log m}{\eps^2}\right)
\]
copies. Since we invoke it $R=O(\log(m/\delta))$ times, the total number of copies is
\[
    O\left( \frac{\log(m)\log(m/\delta)}{\eps^2} \right),
\]
as claimed.
\end{proof}

\end{document}